\documentclass[a4paper,11pt]{article}

\usepackage[T1]{fontenc}
\usepackage[utf8]{inputenc}
\usepackage{fullpage}
\usepackage{times}
\usepackage{microtype}
\usepackage{amsmath,amssymb,amsthm,mathtools}
\usepackage{booktabs}
\usepackage{float}
\usepackage{enumitem}
\usepackage[small]{caption}
\usepackage{authblk}
\usepackage{natbib}
\usepackage[hidelinks]{hyperref}
\usepackage[capitalize,nameinlink]{cleveref}
\usepackage{xcolor}

\newtheorem{theorem}{Theorem}[section]
\newtheorem{corollary}[theorem]{Corollary}
\newtheorem{proposition}[theorem]{Proposition}
\newtheorem{lemma}[theorem]{Lemma}

\theoremstyle{definition}
\newtheorem{definition}[theorem]{Definition}

\newcommand{\cA}{\mathcal{A}}

\newcommand{\IDO}{\mathrm{IDO}}

\DeclareMathOperator{\rank}{rank}
\allowdisplaybreaks
\title{\bf Online Fair Division with Frequency Predictions}

\author[1]{Tzeh Yuan Neoh}
\author[2]{Nicholas Teh}
\affil[1]{Harvard University, USA}
\affil[2]{University of Oxford, UK}

\date{}

\begin{document}

\maketitle

\begin{abstract}
    We study online fair division of indivisible goods among $n$ agents with frequency predictions: before any good arrives, the algorithm learns the multiset of each agent's values for the goods, but neither the arrival order nor which values of different agents belong to the same good. Neoh et al.\ (ICML 2026) showed that such predictions allow every share-based guarantee that is achievable offline, and left open whether a positive approximation of envy-freeness up to one good (EF1) is possible. We show that deterministic online algorithms can guarantee $1/2$-EF1 for three agents and $(3-\sqrt5)/2$-EF1 for four agents, whereas for five or more agents, no online algorithm can guarantee any positive approximation of EF1, even with $n+1$ goods and identical predictions. Together with the known result for two agents, a positive approximation of EF1 is therefore possible if and only if $n\leq4$. We also rule out any positive approximation of EFX for three or more agents, even when the valuations take only two distinct positive values. Next, we use the meta-algorithm of Neoh et al. (ICML 2026) to compute an epistemic EFX (EEFX) allocation in polynomial time for any $n$, and hence a $1/(n-1)$-EF$(n-2)$ allocation for $n\geq3$, whereas algorithms that fix a picking sequence in advance can guarantee a positive approximation of EF$k$ if and only if $n\leq k+2$. Finally, for noisy predictions, we give algorithms that return EEFX allocations when the predictions are exact, and guarantee $(1-2\varepsilon+O(\varepsilon^2))$-EEFX under underestimates and, for $n\geq3$, $(1-2(n-1)\varepsilon+O(\varepsilon^2))$-EEFX under overestimates; the linear terms are the best possible, and larger overestimates rule out any positive approximation for such algorithms.

\end{abstract}

\section{Introduction}

Fair division, the problem of allocating scarce resources among interested agents in a fair manner, has been studied for decades in economics and mathematics, and has received substantial attention in computer science \citep{Moulin2003,AmanatidisEtAl2023}. In many applications, the resources do not become available all at once, but arrive over time and must be allocated upon arrival. For example, a food bank must distribute donated food to charities as the donations come in \citep{AleksandrovEtAl2015}, and an online platform must assign advertising slots to advertisers as users visit the platform \citep{MehtaEtAl2007}. This has led to the study of \emph{online fair division}, where indivisible goods arrive one at a time and each good must be allocated irrevocably when it arrives \citep{AleksandrovWalsh2020}. As in most of the literature, we evaluate fairness on the final allocation, once all goods have been allocated.

In order to reason about fairness, we must define when an allocation is considered to be ``fair''. A natural notion is \emph{envy-freeness}, which requires every agent to like her own bundle at least as much as the bundle of any other agent. However, envy-freeness cannot always be satisfied: if a single valuable good is to be allocated between two agents, the agent who does not receive the good will envy the other agent. A widely studied relaxation is \emph{envy-freeness up to one good (EF1)}, which means that any envy that an agent may have toward another agent can be eliminated by removing some good in the latter agent's bundle \citep{LiptonEtAl2004,Budish2011}. In the offline setting, where all goods are known in advance, an EF1 allocation always exists and can be computed in polynomial time \citep{LiptonEtAl2004,CaragiannisEtAl2019}. In contrast, without any information about future goods, no deterministic online algorithm can guarantee a positive approximation of EF1 \citep{NeohPetersTeh2026}, and the same holds for proportionality up to any fixed number of goods, even when the number of goods is known in advance \citep{neoh2026closinggaps}.

In light of this, \citet{NeohPetersTeh2026} investigated which information about future goods allows online algorithms to achieve fairness. In particular, they considered \emph{frequency predictions}, where the algorithm knows in advance the multiset of values that each agent has for the goods, but neither the order in which the goods arrive nor which values of different agents belong to the same good. They showed that frequency predictions suffice for every guarantee based on \emph{shares}, i.e., thresholds that depend only on an agent's values and the number of agents, as long as the guarantee can be achieved offline. For instance, an agent's \emph{maximin share} \citep{Budish2011} is the best value that she can secure by dividing the goods into $n$ bundles and getting the worst bundle; with frequency predictions, a constant fraction of every agent's maximin share can be guaranteed for any number of agents \citep{AkramiGarg2024,NeohPetersTeh2026}. Envy-based notions such as EF1 are more demanding, since an agent compares her bundle with the bundles of other agents, and these comparisons depend on how the values of different agents are aligned. For example, if two agents each value one good at $10$ and another good at $1$, the predictions do not reveal whether both agents value the same good at $10$. \citet{NeohPetersTeh2026} showed that EFX, which requires any envy that an agent has toward another agent to disappear as soon as any good in the latter agent's bundle is removed, can be guaranteed for two agents, but that no positive approximation of EFX is possible for three or more agents. They left the existence of EF1 allocations for three or more agents as a main open problem: while they conjectured that exact EF1 cannot be guaranteed, they asked whether a positive approximation of EF1 is possible.

In this paper, we answer the latter question for every number of agents. Perhaps surprisingly, the answer depends on the number of agents: a positive approximation of EF1 is possible for up to four agents, but not for five or more agents.

\subsection{Our Results}

We follow \citet{NeohPetersTeh2026} and assume that the agents have additive valuations, and that the arrival order and the values of the goods are chosen by an adaptive adversary subject to the predictions. We describe the model formally in \cref{sec:preliminaries}. Our main result is that, given exact frequency predictions, there exist deterministic online algorithms that guarantee $1/2$-EF1 for three agents and $(3-\sqrt5)/2$-EF1 for four agents, whereas for every $n\geq5$ and every $\alpha>0$, no online algorithm guarantees $\alpha$-EF1 for $n$ agents. Since EFX can be guaranteed for two agents \citep{NeohPetersTeh2026}, a positive approximation of EF1 is therefore possible if and only if $n\leq4$. We establish the result for three agents in \cref{sec:three}, the impossibility for five or more agents in \cref{sec:lower}, and the result for four agents in \cref{sec:four}.

In \cref{sec:three}, we investigate the meta-algorithm of \citet{NeohPetersTeh2026}, which computes an allocation of an offline instance in which all agents rank the goods in the same order, and implements the corresponding picking sequence online. They showed that each agent receives at least her value in the offline allocation. We strengthen this by showing that each good in her offline bundle can be matched to a different good in her online bundle that she values at least as much (\cref{lem:picking-transfer}). As a consequence, if the offline allocation is EFX, then the online allocation is \emph{epistemic EFX (EEFX)} \citep{CaragiannisEtAl2023}, which means that for each agent, the goods outside her bundle can be partitioned into $n-1$ bundles so that she satisfies EFX toward each of them. Hence, an EEFX allocation can be computed online in polynomial time for any number of agents (\cref{thm:eefx}), even though no positive approximation of EFX is possible for three or more agents \citep{NeohPetersTeh2026}. For $n\geq3$, the same algorithm guarantees $\frac1{n-1}$-EF$(n-2)$, where EF$k$ allows up to $k$ goods to be removed from the envied agent's bundle; in particular, this gives $1/2$-EF1 for three agents (\cref{thm:three}). On the other hand, if the picking sequence is chosen from the predictions before any good arrives, the meta-algorithm can guarantee a positive approximation of EF$k$ if and only if $n\leq k+2$ (\cref{thm:fixed-efk}). In particular, this approach cannot give a positive approximation of EF1 for four agents.

Next, in \cref{sec:lower}, we prove the impossibility for five or more agents (\cref{thm:zero}). Our construction uses only $n+1$ goods, the same predictions for all agents, and normalized valuations in which every good has positive value; since every allocation that gives each agent at most one good is EF1, the number of goods cannot be reduced. Moreover, every allocation in our construction is proportional up to one good, so the impossibility holds even though a natural share-based guarantee is satisfied. We also strengthen the impossibility result of \citet{NeohPetersTeh2026} for EFX with three or more agents: no positive approximation is possible even when each agent values one good at a larger value and all other goods at the same smaller value, and even when the algorithm knows that the good of larger value is one of the last two goods (\cref{thm:two-value-efx}).

In \cref{sec:four}, we return to four agents and design an algorithm that does not fix a picking sequence in advance (\cref{thm:four}). The algorithm selects a few goods of large value as they arrive and allocates them so that, for any two agents $i$ and $j$, at most one good in $j$'s bundle is worth more to $i$ than the most valuable selected good that $i$ receives. Each of the remaining goods is allocated upon arrival using concave potential functions computed from the predictions. We do not know whether the factor $(3-\sqrt5)/2\approx0.382$ is the best possible for four agents.

Finally, in \cref{sec:inaccurate}, we consider noisy frequency predictions. Following \citet{NeohPetersTeh2026}, we assign an unused predicted value to each arriving good and run an algorithm for exact predictions on these values. We call an algorithm \emph{consistent} if it returns an EEFX allocation whenever the predictions are exact. For \emph{underestimates}, where each agent's true values are between her predicted values and $1+\varepsilon$ times these values when both are sorted in nonincreasing order, we obtain a consistent algorithm that guarantees $(1-2\varepsilon+O(\varepsilon^2))$-EEFX. For \emph{overestimates}, where the roles of the true and predicted values are reversed, we guarantee $(1-2(n-1)\varepsilon+O(\varepsilon^2))$-EEFX for $n\geq3$. In both cases, the linear term is the best possible for consistent algorithms (\cref{thm:eefx-relative}); hence, the loss from small errors does not depend on $n$ for underestimates, but grows linearly with $n$ for overestimates. Interestingly, while a positive approximation of EEFX is possible for underestimates of any size, overestimates by a factor of $3/2$ already rule out any positive approximation for $n\geq3$ (\cref{thm:eefx-over-impossible}). We also give additive guarantees for EF$k$ in terms of the difference between the true and predicted values (\cref{prop:inaccurate}).
\subsection{Further Related Work}

\paragraph{Online fair division.}
\citet{NeohPetersTeh2026} introduced frequency predictions for online fair division and obtained share-based guarantees; their open problem on EF1 motivates our work. Without any information about the future, the best possible approximation of proportionality up to one good decreases with the number of goods \citep{neoh2026closinggaps,DuSun2026}. \citet{ZhouBaiWu2023} studied the setting in which the total value of each agent is known in advance, while \citet{choo2026approxproponline} considered predictions of the maximum values of the goods, and \citet{neoh2026closinggaps} showed that such predictions allow a constant approximation of proportionality up to one good for any number of agents. Other work has considered restricted classes of valuations \citep{AmanatidisEtAl2025,WangWei2026}, allowed changes to past allocation decisions \citep{HeEtAl2019,HuangTeh2026} or a buffer that delays some of them \citep{AmanatidisEtAl2026Buffers}, imposed budget constraints \citep{CohenTehGoldbergWooldridge2026}, studied randomized algorithms against an oblivious adversary \citep{CohenTehWooldridge2026,ChenLong2026}, and bounded the envy in terms of the number of goods \citep{BenadeEtAl2018,HalpernEtAl2025}. In house allocation, where each agent receives exactly one good, subsidies can be used to eliminate envy both offline \citep{ChooEtAl2024Subsidy} and online \citep{TehEtAl2026House}. Fairness over time has also been studied in temporal fair division, where the allocation should be fair after every round \citep{CooksonEbadianShah2025,ElkindEtAl2025,ChoiLiTeh2026} or the envy is accumulated over all rounds \citep{GoldbergRobinsonTeh2026}, as well as in repeated matching \citep{LimNeohTeh2026Matching}, temporal slot assignment \citep{ElkindKraiczyTeh2022}, and temporal voting \citep{Lackner2020,BulteauEtAl2021,ElkindNeohTeh2024,ZechEtAl2024,ElkindEtAl2025Verifying,ElkindNeohTeh2025,PhillipsEtAl2026,Teh2026Price,Teh2026FJR,HazonSchmerlerTeh2026}; see \citet{ElkindObraztsovaTeh2024} and \citet{Teh2024} for overviews. In our setting, goods are allocated irrevocably, and the fairness of the final allocation is evaluated against an adaptive adversary.

\paragraph{Offline fair division.}
In the offline setting, an EF1 allocation can be computed by envy-cycle elimination or by the round-robin algorithm \citep{LiptonEtAl2004,CaragiannisEtAl2019}; for additive valuations, an allocation with maximum Nash welfare is moreover EF1 and Pareto optimal \citep{CaragiannisEtAl2019}, whereas these two properties are not always compatible for more general valuations \citep{Teh2026Strict}. The maximum Nash welfare rule has also been characterized by the fairness guarantees that it provides \citep{Suksompong2023,YuenSuksompong2023,CelineSuksompongTeh2026}. The reduction to instances in which all agents rank the goods in the same order was introduced by \citet{BouveretLemaitre2016}, and has since been used to approximate the maximin share \citep{BarmanKrishnamurthy2020,GargTaki2021,AkramiGarg2024}. While maximin share allocations do not always exist \citep{KurokawaProcacciaWang2018}, they do exist for a broad class of valuations that generalizes several domains with few distinct values \citep{Teh2026MMS}. Whether EFX allocations always exist is a major open problem; they are known to exist for three agents \citep{ChaudhuryGargMehlhorn2024} and when there are at most three more goods than agents \citep{Mahara2024}, in which case their number \citep{NeohTeh2025} and their welfare \citep{LimNeohTeh2026EFX} have also been studied. \citet{CaragiannisEtAl2023} introduced EEFX as a relaxation of EFX, and computed offline EEFX allocations by finding an EFX allocation of an ordered instance \citep{PlautRoughgarden2020} and matching its goods to the original goods. We prove the corresponding matching property for the online meta-algorithm. The strong EEFX share argument of \citet{AkramiEtAl2026}, together with the share-based result of \citet{NeohPetersTeh2026}, also implies that online EEFX is possible, but does not give a polynomial-time algorithm. Our construction computes both the allocation and the EEFX partitions in polynomial time. Our impossibility result for EFX also separates the online and offline settings: offline, when every value is one of two positive numbers, every allocation with maximum Nash welfare is EFX \citep{AmanatidisEtAl2021}, and an allocation that is both EFX and Pareto optimal can be computed in polynomial time \citep{GargMurhekar2023}. Offline fair division with binary valuations is also well understood: a version of maximum (weighted) Nash welfare is group-strategyproof \citep{HalpernEtAl2020,SuksompongTeh2022,SuksompongTeh2023}, and is the only rule that satisfies EF1, strategyproofness, and a few other natural axioms \citep{BrandlSuksompongTeh2026}.

\paragraph{Frequency predictions and learning-augmented algorithms.}
Predictions of the frequencies of values or types, without the arrival order, have been studied for online knapsack \citep{ImEtAl2021}, online bin packing \citep{AngelopoulosEtAl2023} and bin covering \citep{BergKamali2024}, online bipartite matching \citep{ChooEtAl2024}, and revenue management \citep{BalseiroKroerKumar2026}. In our setting, each agent has her own predicted multiset, so the algorithm does not know which values of different agents belong to the same good; our impossibility results rely on this uncertainty. Similarly, in online fair division with budget constraints, separate predictions of the values and the sizes of the goods do not suffice for strong fairness guarantees \citep{CohenTehGoldbergWooldridge2026}. More broadly, our work belongs to the literature on learning-augmented algorithms, which seeks guarantees that improve with the accuracy of the predictions \citep{PurohitSvitkinaKumar2018,LykourisVassilvitskii2021,MitzenmacherVassilvitskii2022}. \citet{BanerjeeEtAl2022} studied online Nash welfare maximization for divisible goods given predictions of each agent's total value, and \citet{BanerjeeEtAl2023} used such predictions for the online allocation of public goods. For indivisible goods, \citet{MelissourgosProtopapas2026} considered approximate EFX given a predicted value for each arrival position and a known bound on the error, and \citet{TehEtAl2026House} gave learning-augmented algorithms for online house allocation with subsidies. In contrast, our predictions do not reveal the arrival order, and we treat underestimates and overestimates separately; one-sided errors have also been considered for predictions of the maximum values of the goods \citep{neoh2026closinggaps}. We also extend the additive error analysis of \citet{NeohPetersTeh2026} to EF$k$. In a model where past allocation decisions may be changed as long as no affected agent is made worse off, \citet{HuangTeh2026} use predictions of each agent's values in nonincreasing order to carry over offline maximin share guarantees, with an additive loss bounded by the total difference between the sorted true and predicted values.

\section{Preliminaries}\label{sec:preliminaries}

We consider the model of \citet{NeohPetersTeh2026}. Let $N=[n]$ be the set of agents, where $n\geq2$ and we use $[z]$ to denote $\{1,\ldots,z\}$ for each positive integer $z$, and let $G=\{g_1,\ldots,g_m\}$ be the set of $m\geq1$ indivisible goods, indexed in order of arrival. Each agent $i$ has a nonnegative additive valuation $v_i$, i.e., $v_i(S)=\sum_{g\in S}v_i(g)$ for every $S\subseteq G$; the valuations are \emph{normalized} if $v_i(G)=1$ for every agent $i$. An \emph{allocation} $\cA=(A_1,\ldots,A_n)$ is an ordered partition of $G$ into $n$ bundles such that bundle $A_i$ is allocated to agent $i$.

\paragraph{Online setting and frequency predictions.}
Before any good arrives, the algorithm learns, for each agent $i$, the multiset $\{v_i(g):g\in G\}$ of her values for the goods; we call it her \emph{frequency multiset} and write its elements in nonincreasing order as $a_{i1}\geq\cdots\geq a_{im}$. When good $g_t$ arrives, the algorithm observes its value $v_i(g_t)$ for every agent $i$ and must allocate it to an agent irrevocably. The values of the goods are chosen by an adaptive adversary, who observes the previous allocation decisions but must respect the frequency multisets. An algorithm guarantees a fairness property if its final allocation satisfies the property against every such adversary; for randomized algorithms, we require this for every final allocation that can occur.

\subsection{Fairness Notions}

An allocation $\cA$ is said to be \emph{envy-free up to one good (EF1)} if for all distinct agents $i,j$ with $A_j\neq\varnothing$, there exists a good $g\in A_j$ such that $v_i(A_i)\geq v_i(A_j\setminus\{g\})$, and \emph{envy-free up to any good (EFX)} if $v_i(A_i)\geq v_i(A_j\setminus\{g\})$ for all distinct agents $i,j$ and every good $g\in A_j$. We include goods of zero value in these comparisons.\footnote{\citet{CaragiannisEtAl2019} required the EFX comparison only for goods of positive value. Since all values in our impossibility results are positive, these results hold under either definition.}

More generally, for a positive integer $k$ and a bundle $S\subseteq G$, let $D_i^{(k)}(S)$ denote agent $i$'s value for $S$ after removing her $k$ most valuable goods in $S$, where $D_i^{(k)}(S)=0$ if $|S|\leq k$; we write $D_i(S)=D_i^{(1)}(S)$. For $\alpha\in[0,1]$, an allocation is \emph{$\alpha$-EF$k$} if $v_i(A_i)\geq\alpha D_i^{(k)}(A_j)$ for all distinct agents $i,j$, and \emph{$\alpha$-EFX} if $v_i(A_i)\geq\alpha v_i(A_j\setminus\{g\})$ for all distinct agents $i,j$ and every $g\in A_j$. In particular, $1$-EF$1$ is equivalent to EF1. An algorithm guarantees a \emph{positive approximation} of EF$k$ if it guarantees $\alpha$-EF$k$ for some $\alpha>0$ that may depend on $n$ and $k$, but not on $m$ or the valuations.

\begin{definition}[\citealp{CaragiannisEtAl2023}]
Let $\alpha\in[0,1]$. A bundle $S\subseteq G$ is \emph{$\alpha$-EEFX for agent $i$} if there exists a partition $(Y_1,\ldots,Y_{n-1})$ of $G\setminus S$ such that $v_i(S)\geq\alpha\, v_i(Y_p\setminus\{g\})$ for every $p\in[n-1]$ and every $g\in Y_p$. An allocation $\cA$ is \emph{$\alpha$-EEFX} if $A_i$ is $\alpha$-EEFX for agent $i$ for every $i\in N$. We write EEFX for $1$-EEFX.
\end{definition}

In other words, an agent's bundle is EEFX for her if she would satisfy EFX toward every other agent after a suitable redistribution of the goods that she does not receive. The partition may be different for each agent, and some of its parts may be empty. Whether a bundle is EEFX for an agent depends only on the bundle and her valuation, and not on how the remaining goods are allocated. For two agents, EEFX is equivalent to EFX.

We also compare bundles via \emph{stochastic dominance} \citep[Definition~2]{BarmanEtAl2025}. For $S,S'\subseteq G$, we write $S'\succeq_i S$ if for every $t\geq0$, the number of goods in $S'$ that agent $i$ values at least $t$ is at least the number of such goods in $S$. Equivalently, $|S'|\geq|S|$, and for every $r\leq|S|$, agent $i$'s $r$th largest value for a good in $S'$ is at least her $r$th largest value for a good in $S$. In other words, the goods in $S$ can be matched to different goods in $S'$ without decreasing their values to agent $i$.

\subsection{The Meta-Algorithm}\label{sec:transfer}

We now describe the meta-algorithm of \citet[Algorithm~2]{NeohPetersTeh2026}. An instance is said to be an \emph{identical ordering (IDO)} instance if all agents rank the goods in the same order. Given the frequency predictions, consider the IDO instance in which agent $i$ values each good $g_r$ at $a_{ir}$; we denote this valuation by $v_i^{\IDO}$. A \emph{picking sequence} \citep{KohlerChandrasekaran1971,BouveretLang2011} $\pi=(\pi_1,\ldots,\pi_m)\in N^m$ specifies that in turn $r$, agent $\pi_r$ picks a most valuable remaining good. In the IDO instance, if ties are broken in favor of smaller indices, the good picked in turn $r$ is $g_r$, so the picking sequence results in the allocation $(B_1,\ldots,B_n)$ with $B_i=\{g_r:\pi_r=i\}$. We refer to $B_i$ as agent $i$'s \emph{IDO bundle}.

The meta-algorithm implements a given picking sequence online as follows. When a good $g$ arrives, it removes $v_i(g)$ from the remaining predicted values of each agent $i$. Let $s$ be the number of remaining turns. The algorithm simulates the remaining picking sequence on $g$ together with $s-1$ hypothetical goods, where each agent values the hypothetical goods at her remaining predicted values in nonincreasing order. In the simulation, each turn picks a most valuable remaining good, breaking ties in favor of $g$ and then in favor of hypothetical goods with smaller indices. The algorithm allocates $g$ to the agent whose turn picks $g$ in the simulation, and deletes this turn from the picking sequence.

\citet{NeohPetersTeh2026} showed that the final allocation satisfies $|A_i|=|B_i|$ and $v_i(A_i)\geq v_i^{\IDO}(B_i)$ for every agent $i$. We strengthen this guarantee as follows: each good in $B_i$ can be matched to a different good in $A_i$ that agent $i$ values at least as much, and conversely, the goods outside $A_i$ are worth no more to her than the goods outside $B_i$ in the IDO instance.

\begin{lemma}\label[lemma]{lem:picking-transfer}
Given any picking sequence, the meta-algorithm returns an allocation $\cA$ such that the following hold for every agent $i$.
\begin{enumerate}[label=\textup{(\roman*)},leftmargin=2em]
\item There exists a bijection $\varphi_i:A_i\to B_i$ such that $v_i(g)\geq v_i^{\IDO}(\varphi_i(g))$ for every $g\in A_i$. In particular, $|A_i|=|B_i|$ and $v_i(A_i)\geq v_i^{\IDO}(B_i)$.
\item For every $r$, agent $i$'s $r$th largest value for a good in $G\setminus A_i$ is at most her $r$th largest value under $v_i^{\IDO}$ for a good in $G\setminus B_i$.
\end{enumerate}
\end{lemma}

Intuitively, if the goods yet to arrive formed an IDO instance with the remaining predicted values, then a turn in position $\ell$ of the remaining picking sequence would give its agent her $\ell$th largest remaining predicted value. We prove part~(i) in \cref{app:picking-transfer} by showing that this value never decreases until the turn is deleted, and that the good allocated at that time is worth at least this value to the agent. Part~(ii) follows from part~(i), since agent $i$'s true values and her values under $v_i^{\IDO}$ form the same multiset.

\section{Picking Sequences, EF\texorpdfstring{$k$}{k}, and Three Agents}\label{sec:three}

In this section, we investigate the guarantees that can be obtained via the meta-algorithm. We first establish two properties of EEFX, and then use them to show that the meta-algorithm yields EEFX for any number of agents and $1/2$-EF1 for three agents. We also show that algorithms that fix a picking sequence in advance cannot guarantee a positive approximation of EF$k$ when $n\geq k+3$.

\subsection{Two Properties of EEFX}

Our first lemma shows that EEFX is preserved under stochastic dominance.

\begin{lemma}\label[lemma]{lem:eefx-fosd}
Let $i$ be an agent and $S,S'\subseteq G$. If $S$ is EEFX for agent $i$ and $S'\succeq_i S$, then $S'$ is EEFX for agent $i$. Moreover, if $|S'|=|S|$, then for every partition showing that $S$ is EEFX for agent $i$, there is a partition with the same part sizes showing that $S'$ is EEFX for agent $i$.
\end{lemma}

\begin{proof}
Let $(Y_1,\ldots,Y_{n-1})$ be a partition of $G\setminus S$ showing that $S$ is EEFX for agent $i$. Since $S'\succeq_i S$, for every $t\geq0$, the set $G\setminus S'$ contains at most as many goods that agent $i$ values at least $t$ as $G\setminus S$ does. Hence, if we list the goods in $G\setminus S'$ and in $G\setminus S$ in nonincreasing order of their values to agent $i$ and match them in this order, we obtain an injection $\psi:G\setminus S'\to G\setminus S$ such that $v_i(g)\leq v_i(\psi(g))$ for every $g\in G\setminus S'$. For each $p\in[n-1]$, let $Y'_p=\{g\in G\setminus S':\psi(g)\in Y_p\}$; these sets form a partition of $G\setminus S'$. For every $p$ and every $g\in Y'_p$, since $\psi$ is injective and values are nonnegative, we have
\[
v_i(Y'_p\setminus\{g\})\leq v_i(Y_p\setminus\{\psi(g)\})\leq v_i(S)\leq v_i(S'),
\]
where the last inequality holds because $S'\succeq_i S$. Hence $S'$ is EEFX for agent $i$. If $|S'|=|S|$, then $\psi$ is a bijection, so $|Y'_p|=|Y_p|$ for every $p$.
\end{proof}

Our second lemma shows that EEFX implies an approximation of EF$(n-2)$ when every agent receives at least one good.

\begin{lemma}\label[lemma]{lem:eefx-efk}
Let $n\geq3$, and let $\cA$ be an allocation in which every bundle is nonempty. If $A_i$ is EEFX for agent $i$, then $v_i(A_i)\geq\frac1{n-1}D_i^{(n-2)}(A_j)$ for every agent $j\neq i$. In particular, if $\cA$ is EEFX, then it is $\frac1{n-1}$-EF$(n-2)$.
\end{lemma}

\begin{proof}
Fix an agent $j\neq i$, and let $(Y_1,\ldots,Y_{n-1})$ be a partition of $G\setminus A_i$ showing that $A_i$ is EEFX for agent $i$. Since $n\geq3$ and every bundle is nonempty, some good lies outside $A_i\cup A_j$, and the part containing this good is not contained in $A_j$. Hence at most $n-2$ nonempty parts are contained in $A_j$. Let $X$ consist of one good from each of these parts, so that $|X|\leq n-2$. For every $p$, the set $(A_j\setminus X)\cap Y_p$ is either empty or contained in $Y_p\setminus\{g\}$ for some $g\in Y_p$, so its value to agent $i$ is at most $v_i(A_i)$. Therefore,
\begin{align*}
D_i^{(n-2)}(A_j)\leq v_i(A_j\setminus X)
&=\sum_{p=1}^{n-1}v_i((A_j\setminus X)\cap Y_p) \leq(n-1)\,v_i(A_i). \qedhere
\end{align*}
\end{proof}

The assumption $n\geq3$ in \cref{lem:eefx-efk} cannot be dropped, since for two agents, EEFX is equivalent to EFX, which does not imply envy-freeness. The assumption that every bundle is nonempty ensures that some good lies outside the two bundles being compared.

\subsection{Guarantees of the Meta-Algorithm}

For an IDO instance, an EFX allocation can be computed in polynomial time by envy-cycle elimination, allocating the goods in their common order \citep{LiptonEtAl2004,PlautRoughgarden2020,BarmanKrishnamurthy2020}. When $m\geq n$, we first give the goods $g_1,\ldots,g_n$ to different agents, which results in an EFX allocation, and then continue with envy-cycle elimination; this preserves EFX and ensures that every agent receives at least one good. Our algorithm computes such an allocation $(B_1,\ldots,B_n)$ of the IDO instance given by the predictions, breaking ties deterministically, and runs the meta-algorithm with the picking sequence $\pi$ given by $\pi_r=i$ for $g_r\in B_i$.

\begin{theorem}\label{thm:eefx}
For any number of agents, there exists a deterministic polynomial-time online algorithm that, given exact frequency predictions, returns an EEFX allocation. Moreover, if $m\geq n$, then every bundle in this allocation is nonempty, and for every agent $i$, the partition showing that $A_i$ is EEFX for agent $i$ can be chosen to have parts of sizes $|A_j|$ for $j\neq i$.
\end{theorem}

\begin{proof}
If $m<n$, we give the goods to different agents; the resulting allocation is EFX, and therefore EEFX. Suppose that $m\geq n$, and let $(B_1,\ldots,B_n)$ and $\cA$ be the allocations described above. Fix an agent $i$. Since her true values are $a_{i1},\ldots,a_{im}$, there exists a bijection $\sigma:G\to G$ such that $v_i(\sigma(g_r))=a_{ir}=v_i^{\IDO}(g_r)$ for every $r\in[m]$. Since $(B_1,\ldots,B_n)$ is EFX in the IDO instance, the sets $\sigma(B_j)$ for $j\neq i$ form a partition of $G\setminus\sigma(B_i)$ showing that $\sigma(B_i)$ is EEFX for agent $i$. By \cref{lem:picking-transfer}(i), we have $A_i\succeq_i\sigma(B_i)$ and $|A_i|=|\sigma(B_i)|$. Hence, by \cref{lem:eefx-fosd}, $A_i$ is EEFX for agent $i$, and the corresponding partition can be chosen to have parts of sizes $|B_j|=|A_j|$ for $j\neq i$. Finally, every bundle is nonempty, since $|A_j|=|B_j|\geq1$ for every $j$.
\end{proof}

\Cref{thm:eefx} stands in contrast to EFX, for which no positive approximation can be guaranteed when $n\geq3$ \citep{NeohPetersTeh2026}. The difference is that \cref{lem:picking-transfer} gives each agent a bundle at least as good as her IDO bundle, but says nothing about how the goods outside her bundle are divided among the other agents; EEFX, unlike EFX, does not depend on this division.

Combining \cref{thm:eefx} with \cref{lem:eefx-efk} yields the following corollary.

\begin{corollary}\label{thm:three}
For any $n\geq3$, there exists a deterministic polynomial-time online algorithm that, given exact frequency predictions, returns a $\frac1{n-1}$-EF$(n-2)$ allocation. In particular, $1/2$-EF1 can be guaranteed for three agents.
\end{corollary}

\begin{proof}
If $m\geq n$, the allocation in \cref{thm:eefx} is EEFX and every bundle is nonempty, so the claim follows from \cref{lem:eefx-efk}. If $m<n$, giving the goods to different agents results in an EF1 allocation, which is also EF$(n-2)$.
\end{proof}

\subsection{Limits of Fixed Picking Sequences}

\Cref{thm:eefx,thm:three} give a positive approximation of EF$k$ whenever $n\leq k+2$. We now show that this is the best possible for algorithms that implement a picking sequence chosen in advance.

\begin{theorem}\label{thm:fixed-efk}
Let $n\geq2$ and $k\geq1$, and consider algorithms that choose a picking sequence based on the frequency predictions before any good arrives, possibly at random, and implement it using the meta-algorithm. Such an algorithm can guarantee a positive approximation of EF$k$ if and only if $n\leq k+2$. Moreover, for $n\geq k+3$, the impossibility holds even with $nk+1$ goods, identical predictions for all agents, and normalized valuations in which each agent values all goods at distinct positive values.
\end{theorem}

The positive direction follows from \cref{thm:eefx} for $n=2$, since EEFX is equivalent to EFX for two agents, and from \cref{thm:three} for $3\leq n\leq k+2$, since $D_i^{(k)}(S)\leq D_i^{(n-2)}(S)$ for every bundle $S$ whenever $k\geq n-2$. For the impossibility, we force an agent to receive goods of small value until all of her turns have been deleted, and then give the next $k+1$ goods, which she values highly, to another agent. The proof is given in \cref{app:fixed-efk}. The number of goods in the impossibility cannot be reduced: if $m\leq nk$, a picking sequence that gives each agent at most $k$ turns results in an EF$k$ allocation.

For four agents and $k=1$, the impossibility in \cref{thm:fixed-efk} holds already with five goods. Hence, a positive approximation of EF1 for four agents requires an algorithm that does not fix its picking sequence in advance. We present such an algorithm in \cref{sec:four}.

\section{Impossibility Results}\label{sec:lower}

In this section, we prove impossibility results that apply to all online algorithms given frequency predictions. We first show that no positive approximation of EF1 is possible for five or more agents, and then show that no positive approximation of EFX is possible for three or more agents, even when the valuations take only two distinct positive values.

\subsection{EF1 for Five or More Agents}

Our construction for EF1 uses $n+1$ goods and gives all agents the same frequency multiset. After the first two goods have been allocated, the adversary chooses the values of all remaining goods. At a high level, if every agent receives a good, then exactly one agent receives two goods, and every other agent must value her own good at least $\alpha$ times as much as the less valuable of these two goods. If the first two goods go to the same agent, the adversary makes the last good worth little to all other agents. Otherwise, the recipients of the first two goods each hold a good worth much less to her than her largest value; the adversary then makes each later good highly valuable to at least one of them, and the last good worth little to all other agents.

\begin{theorem}\label{thm:zero}\label{cor:zero-factor}
For every $n\geq5$ and every $\alpha\in(0,1]$, no online algorithm guarantees $\alpha$-EF1 for $n$ agents given exact frequency predictions. This holds even with $n+1$ goods, identical predictions for all agents, and normalized valuations in which every good has positive value.
\end{theorem}

\begin{proof}
Fix an online algorithm, and let $K>\max\{2,1/\alpha\}$. Every agent's frequency multiset consists of $n-2$ copies of $K^3$ together with the values $K^2$, $K$, and $1$. We show that the final allocation contains distinct agents $i,j$ with $D_i(A_j)>0$ and $v_i(A_i)\leq D_i(A_j)/K$; since $1/K<\alpha$, the allocation is then not $\alpha$-EF1.

Since all values are positive, if some bundle $A_i$ is empty, then some other bundle $A_j$ contains at least two goods, so $D_i(A_j)>0=v_i(A_i)$, as desired. Hence, we may assume that every bundle is nonempty. As there are $n+1$ goods, exactly one agent receives two goods, and every other agent receives one good.

The first good $g_1$ has value $K^2$ to every agent; we relabel its recipient as agent~$1$. The second good $g_2$ has value $K^3$ to agent~$1$ and value $K$ to every other agent. The values of the remaining goods depend on the recipient of $g_2$.

\paragraph{Case 1: $g_2$ is allocated to agent $1$.}
Every agent $i\neq1$ values $g_3,\ldots,g_n$ at $K^3$ and $g_{n+1}$ at $1$, whereas agent~$1$ values these goods at her remaining predicted values in an arbitrary order. Since $A_1=\{g_1,g_2\}$ is the bundle with two goods, $g_{n+1}$ is the only good of some agent $i\neq1$. Then $v_i(A_i)=1$ and $D_i(A_1)=\min\{K^2,K\}=K$, as desired.

\paragraph{Case 2: $g_2$ is allocated to another agent.}
We relabel this agent as agent~$2$. The values of $g_3,\ldots,g_{n+1}$ are given in \cref{tab:impossibility-values}; each row is consistent with the predictions.

\begin{table}[t]
\caption{Values in Case~2 of the proof of \cref{thm:zero}, with the goods listed in order of arrival. Each of the goods $g_6,\ldots,g_n$ has the value in the corresponding column; this column is absent when $n=5$.}
\label{tab:impossibility-values}
\begin{center}
\setlength{\tabcolsep}{4.5pt}
\begin{tabular}{@{}lccccccc@{}}
\toprule
Agent & $g_1$ & $g_2$ & $g_3$ & $g_4$ & $g_5$ & $g_6,\ldots,g_n$ & $g_{n+1}$ \\
\midrule
$1$ & $K^2$ & $K^3$ & $K$ & $1$ & $K^3$ & $K^3$ & $K^3$ \\
$2$ & $K^2$ & $K$ & $K^3$ & $K^3$ & $1$ & $K^3$ & $K^3$ \\
$i\notin\{1,2\}$ & $K^2$ & $K$ & $K^3$ & $K^3$ & $K^3$ & $K^3$ & $1$ \\
\bottomrule
\end{tabular}
\end{center}
\end{table}

Each of the goods $g_3,\ldots,g_n$ is worth $K^3$ to at least one of agents~$1$ and~$2$, and $g_{n+1}$ is worth $K^3$ to both of them. This is where we use the assumption $n\geq5$: the three goods $g_3,g_4,g_5$ allow agent~$1$'s remaining small values $K$ and $1$ and agent~$2$'s remaining small value $1$ to be placed on different goods.

Let $j$ be the agent who receives two goods. Suppose first that $g_{n+1}\notin A_j$. Since agents~$1$ and~$2$ already hold $g_1$ and $g_2$, respectively, $g_{n+1}$ is the only good of some agent $i\in N\setminus\{1,2\}$. This agent values $g_{n+1}$ at $1$ and every other good at least $K$, so $v_i(A_i)=1$ and $D_i(A_j)\geq K$, as desired.

Suppose now that $g_{n+1}\in A_j$. If $j=1$, then $A_1=\{g_1,g_{n+1}\}$ and $A_2=\{g_2\}$, so $v_2(A_2)=K$ and $D_2(A_1)=\min\{K^2,K^3\}=K^2$. If $j=2$, then $A_2=\{g_2,g_{n+1}\}$ and $A_1=\{g_1\}$, so $v_1(A_1)=K^2$ and $D_1(A_2)=K^3$. Finally, if $j\notin\{1,2\}$, then the other good of agent $j$ is among $g_3,\ldots,g_n$, so by \cref{tab:impossibility-values}, some agent $i\in\{1,2\}$ values both goods in $A_j$ at $K^3$. Since agents~$1$ and~$2$ receive only $g_1$ and $g_2$, respectively, we have $v_i(A_i)\leq K^2$ and $D_i(A_j)=K^3$. In all cases, the desired agents exist.

To obtain normalized valuations, we divide all predicted and true values by $(n-2)K^3+K^2+K+1$ from the start; this does not affect any of the comparisons. Finally, since the values of future goods depend only on the allocation decisions observed so far, the argument applies to every possible outcome of a randomized algorithm.
\end{proof}

\paragraph{Comparison with share-based fairness.}
Every allocation of the instances in the proof of \cref{thm:zero} is \emph{proportional up to one good} \citep{ConitzerFreemanShah2017}, meaning that each agent can reach $1/n$ of her total value by adding at most one good from outside her bundle. Indeed, the total value $(n-2)K^3+K^2+K+1$ of each agent is less than $nK^3$, and each agent either already has a good of value $K^3$ or can add one. This illustrates the difference between guarantees on an agent's own value and envy-based guarantees; the former, and in particular approximations of proportionality up to one good, have been the focus of much recent work on online fair division \citep{choo2026approxproponline,neoh2026closinggaps,CohenTehWooldridge2026,DuSun2026}.

\subsection{EFX with Two Positive Values}\label{sec:two-value-efx}

We now turn to EFX. \citet[Section~5]{NeohPetersTeh2026} showed that no positive approximation of EFX is possible for $n\geq3$ given frequency predictions. We strengthen their result to a restricted class of valuations, in which each agent values one good at a larger value and all other goods at the same smaller value; moreover, the good of larger value is one of the last two goods.

\begin{theorem}\label{thm:two-value-efx}
For every $n\geq3$ and every $\alpha\in(0,1]$, no online algorithm guarantees $\alpha$-EFX for $n$ agents given exact frequency predictions. This holds even with identical predictions for all agents and normalized valuations that take only two positive values, and even if the algorithm knows in advance that every agent values each of the first $m-2$ goods at the smaller value and exactly one of the last two goods at the larger value.
\end{theorem}

We prove \cref{thm:two-value-efx} in \cref{app:two-value-efx}. In our construction, each agent values one good at $1$ and all other goods at a small value $\ell$, and each of the first $m-2$ goods has value $\ell$ to every agent. Every agent who receives one of these goods then values $g_{m-1}$ at $1$, and every other agent values $g_m$ at $1$. If $g_{m-1}$ is allocated to an agent who received one of the first $m-2$ goods, then any other such agent has a small value for her own bundle, but values the bundle containing $g_{m-1}$ at least $1$ after removing a good of value $\ell$. The remaining cases are similar, except when $g_{m-1}$ and $g_m$ are allocated to two different agents who received none of the first $m-2$ goods. In this case, some agent receives many of the first $m-2$ goods, whereas the agent who receives $g_{m-1}$ values her bundle at only $\ell$.

\section{Four Agents: \texorpdfstring{$(3-\sqrt5)/2$}{0.3820}-EF1 Guarantee}\label{sec:four}

We now present our algorithm for four agents. Let $C=(3+\sqrt5)/2$, so that $1/C=(3-\sqrt5)/2$; we show that the algorithm guarantees $1/C$-EF1. The algorithm selects and allocates a few goods of large value as they arrive so that, for any two agents $i$ and $j$, at most one good in $A_j$ is worth more to agent $i$ than the most valuable selected good that she receives. Each good that is not selected is allocated upon arrival using concave potential functions computed from the predictions.

If $m<4$, giving the goods to different agents results in an EF1 allocation, so assume that $m\geq4$. We first cap each agent's values at her second largest predicted value, i.e., we replace $v_i(g)$ by $\min\{v_i(g),a_{i2}\}$ both in the predictions and for the arriving goods. This does not change $D_i(S)$ for any bundle $S$, since only a good of largest value in $S$ can be affected and this good is removed when computing $D_i(S)$; moreover, it can only decrease the value of an agent's own bundle. Hence, a guarantee for the capped values also holds for the original values. From now on, $v_i$ and $a_{ik}$ refer to the capped values, so $a_{i1}=a_{i2}$. For each agent $i$, we associate each arriving good $g$ with an index $k$ such that $v_i(g)=a_{ik}$, where goods of equal value are associated with increasing indices in order of arrival; we then call $g$ agent $i$'s \emph{$k$th occurrence}. Thus, her first four occurrences correspond to her four largest capped values, with ties resolved by arrival order.

\begin{lemma}\label[lemma]{lem:four-selection}
There exists a deterministic online rule that selects some of the arriving goods and allocates each selected good upon its arrival, such that each agent receives one or two selected goods, including at least one of her first four occurrences. Moreover, let $P_i$ be the set of selected goods that agent $i$ receives, and let $h_i=\max_{g\in P_i}v_i(g)$. Then, for all distinct agents $i$ and $j$, at most one good in $A_j$ has value larger than $h_i$ to agent $i$, regardless of how the goods that are not selected are allocated.
\end{lemma}

We prove \cref{lem:four-selection} in \cref{app:four-selection}. The lemma allows us to also cap each agent's values at her fourth largest predicted value. Let $w_i(g)=\min\{v_i(g),a_{i4}\}$ for each good $g$, extended additively, and for a bundle $S$, let $D_i^w(S)$ denote the value of $S$ under $w_i$ after removing a most valuable good under $w_i$, where $D_i^w(\varnothing)=0$. We claim that $D_i(A_j)\leq D_i^w(A_j)+h_i-a_{i4}$ and $v_i(A_i)\geq w_i(A_i)+h_i-a_{i4}$ for all distinct agents $i,j$. The second inequality holds because $A_i$ contains a selected good of value $h_i\geq a_{i4}$ to agent $i$. For the first inequality, remove from $A_j$ a good of largest value to agent $i$; this good also has the largest value under $w_i$. By \cref{lem:four-selection}, each remaining good in $A_j$ has value at most $h_i$ to agent $i$. If $h_i=a_{i4}$, no remaining good in $A_j$ has value larger than $a_{i4}$ to agent $i$. Otherwise, at most three goods have value larger than $a_{i4}$ to agent $i$, and one of them belongs to $A_i$, so at most one remaining good in $A_j$ has value larger than $a_{i4}$, and this value is at most $h_i$. This proves the claim. Consequently, for any $C\geq1$, if $D_i^w(A_j)\leq Cw_i(A_i)$, then $D_i(A_j)\leq C(w_i(A_i)+h_i-a_{i4})\leq Cv_i(A_i)$.

\paragraph{Allocating the unselected goods.}
For each agent $i$, let $\ell_i=w_i(A_i)-\max_{g\in P_i}w_i(g)$ be the value of her current bundle under $w_i$ without a most valuable selected good, where the maximum is zero if $P_i$ is empty. This quantity never decreases: a good that is not selected increases it by its value under $w_i$, and a selected good increases it by the smaller of its value under $w_i$ and the previous maximum. At the end, $w_i(A_i)=a_{i4}+\ell_i$, since $P_i$ contains a good of value at least $a_{i4}$ to agent $i$. We compute from the predictions
\begin{equation*}
F_i=\sum_{k=8}^m a_{ik},\quad
x_i=\frac{\max\{0,F_i-3(C-1)a_{i4}\}}{3C+1}, \quad
T_i=2a_{i4}+\left(3+\frac1C\right)x_i,
\end{equation*}
where $F_i=0$ if $m<8$. The sum starts from the eighth occurrence because we will remove at most seven selected goods. We will show that $\ell_i\geq x_i$ at the end; otherwise, the goods that are not selected and are allocated to the other three agents could not account for the value $F_i$. If $\ell_i\geq T_i$, we compare agent $i$'s value for her bundle with her total value. If $\ell_i\leq T_i$, we use the following potential functions.

\begin{lemma}\label[lemma]{lem:four-potentials}
For every agent $i$ with $a_{i4}>0$, there exists a nondecreasing, concave, piecewise-linear function $\Phi_i$ that depends only on the predictions, such that $\Phi_i(0)=0$, $\Phi_i(u)\leq1$ for all $u\geq0$, and $\Phi_i(u+a_{i4})-\Phi_i(u)\geq a_{i4}/((C-1)a_{i4}+C\max\{x_i,u\})$ for all $u\in[0,T_i]$.
\end{lemma}

We construct these functions in \cref{app:four-potentials}. For an agent $i$ with $a_{i4}=0$, let $\Phi_i\equiv0$. The algorithm applies the rule of \cref{lem:four-selection} to each arriving good $g$. If $g$ is selected, it allocates $g$ to the agent chosen by this rule and adds $g$ to her set of selected goods. Otherwise, it allocates $g$ to an agent $i$ maximizing $\Phi_i(\ell_i+w_i(g))-\Phi_i(\ell_i)$, breaking ties in a fixed order. After either allocation, it updates the recipient's value of $\ell_i$. Intuitively, since each $\Phi_i$ is concave, the algorithm tends to favor agents whose bundles are currently worth little to them. Moreover, since the potentials are bounded, the other agents can collect only a limited amount of value according to agent $i$ while her own value is small; once her own value is large, her total value yields the required bound.

\begin{theorem}\label{thm:four}
For four agents, there exists a deterministic polynomial-time online algorithm that, given exact frequency predictions, returns a $(3-\sqrt5)/2$-EF1 allocation.
\end{theorem}

We prove \cref{thm:four} in \cref{app:four-proof}, and outline the argument here. By the discussion above, it suffices to show that $D_i^w(A_j)\leq Cw_i(A_i)$ for all distinct agents $i,j$ with $a_{i4}>0$, where $\ell_i$ now denotes its final value. Consider a good $g\in A_j\setminus P_j$. When $g$ arrives, the increase in agent $j$'s potential is at least the increase in agent $i$'s potential that $g$ would cause. Since $\ell_i$ never decreases, $\Phi_i$ is concave, and agent $j$'s potential is at most $1$, \cref{lem:four-potentials} implies that if $\ell_i\leq T_i$, then
\[
w_i(A_j\setminus P_j)\leq(C-1)a_{i4}+C\max\{x_i,\ell_i\}.
\]
Applying this bound to each of the other three agents gives $\ell_i\geq x_i$, as mentioned above. If $\ell_i\leq T_i$, then since $A_j$ contains at most two selected goods, we obtain
\[
D_i^w(A_j)\leq a_{i4}+w_i(A_j\setminus P_j)\leq C(a_{i4}+\ell_i)=Cw_i(A_i).
\]
If $\ell_i\geq T_i$, we show that $D_i^w(A_j)$ is at most agent $i$'s total value for the goods other than her first four occurrences, and the choice of $T_i$ ensures that this total value is at most $C(a_{i4}+T_i)\leq Cw_i(A_i)$.

The value $C=(3+\sqrt5)/2$ is the smallest possible for the potential requirements in \cref{lem:four-potentials}. When $x_i=0$, we have $T_i=2a_{i4}$, so summing the inequalities at $u=0,a_{i4},2a_{i4}$ gives
\[
\frac1{C-1}+\frac1{2C-1}+\frac1{3C-1}\leq\Phi_i(3a_{i4})\leq1,
\]
with equality at the chosen value of $C$. We do not know whether the factor $(3-\sqrt5)/2$ is the best possible for four agents.

\section{Learning-Augmented Guarantees under Noisy Frequency Predictions}\label{sec:inaccurate}

We now consider noisy frequency predictions, as in \citet[Section~5.1]{NeohPetersTeh2026}. For each agent $i$, the algorithm is given a multiset of $m$ predicted values, which we write in nonincreasing order as $\widehat a_{i1}\geq\cdots\geq\widehat a_{im}$; as before, $a_{i1}\geq\cdots\geq a_{im}$ denote her true values in nonincreasing order. In particular, the algorithm knows the number of goods, but not the total value of each agent. We say that an algorithm is \emph{consistent} if it returns an EEFX allocation whenever the predictions are exact, i.e., $\widehat a_{ir}=a_{ir}$ for all $i\in N$ and $r\in[m]$.

Following \citet{NeohPetersTeh2026}, we use an online \emph{instantiation rule}, which assigns to each arriving good $g$ an unused predicted value $\widetilde v_i(g)$ of each agent $i$. The resulting \emph{predicted valuation} $\widetilde v_i$ uses each predicted value of agent $i$ exactly once. We run an algorithm for exact predictions on the predicted valuations, and allocate each good to the agent chosen by this algorithm. We first give multiplicative guarantees for EEFX, then show that larger overestimates rule out any positive approximation, and finally give additive guarantees for EF$k$.

\subsection{Relative Prediction Error}
\label{sec:eefx-relative}

We say that agent $i$'s predictions are \emph{underestimates with error factor $\rho_i\geq1$} if $\widehat a_{ir}\leq a_{ir}\leq\rho_i\widehat a_{ir}$ for every $r\in[m]$, and \emph{overestimates with error factor $\rho_i$} if $a_{ir}\leq\widehat a_{ir}\leq\rho_i a_{ir}$ for every $r\in[m]$. If the total value of each agent were known and equal to her total predicted value, either condition would imply that the predictions are exact.

For underestimates, the instantiation rule assigns to an arriving good the largest unused predicted value that is at most its true value, or the smallest unused predicted value if no such value exists. For overestimates, it assigns the smallest unused predicted value that is at least the true value, or the largest unused predicted value if no such value exists. In both cases, ties are broken deterministically, and we run the algorithm of \cref{thm:eefx} on the predicted valuations. Note that neither rule uses the error factors.

\begin{theorem}\label{thm:eefx-relative}\label{thm:eefx-under}\label{thm:eefx-over}\label{thm:eefx-smooth-optimal}
For any number of agents, the two instantiation rules above yield deterministic polynomial-time consistent algorithms with the following guarantees for every agent $i$, regardless of the predictions for the other agents.
\begin{enumerate}[label=\textup{(\roman*)},leftmargin=2em]
\item If agent $i$'s predictions are underestimates with error factor $\rho_i$, then $A_i$ is $\frac{1}{2\rho_i-1}$-EEFX for her.
\item If agent $i$'s predictions are overestimates with error factor $\rho_i$, then $A_i$ is $\alpha$-EEFX for her with $\alpha=\max\{0,(2(n-1)-(2n-3)\rho_i)/\rho_i^2\}$. If $n\geq3$, this also holds with $\alpha=\max\{0,1-2(n-1)(\rho_i-1)-(n-1)(2n-1)(\rho_i-1)^2\}$.
\end{enumerate}
On the other hand, for every $\rho>1$, no consistent algorithm guarantees $\beta$-EEFX for any $\beta>1/(\rho+1-1/\rho)$ when all agents' predictions are underestimates with error factor $\rho$. Likewise, for every $1<\rho<2^{1/(2n-2)}$, no consistent algorithm guarantees $\beta$-EEFX for any $\beta>\rho^{-2(n-1)}$ when all agents' predictions are overestimates with error factor $\rho$. These impossibility results hold even with identical predictions and identical true valuations in which every good has positive value.
\end{theorem}

For fixed $n\geq3$ and $\rho_i=\rho=1+\varepsilon$, the guarantees in \cref{thm:eefx-relative} are $1-2\varepsilon+O(\varepsilon^2)$ for underestimates and $1-2(n-1)\varepsilon+O(\varepsilon^2)$ for overestimates, and the impossibility results show that the linear terms cannot be improved. For overestimates, this uses the second bound in part~(ii). Hence, the best possible first-order loss is independent of $n$ for underestimates, but grows linearly with $n$ for overestimates. For underestimates, the linear term is also the best possible for $n=2$.

We prove \cref{thm:eefx-relative} in \cref{app:eefx-inaccurate}. For underestimates, agent $i$ values her bundle at least as much as its predicted value, and comparing the number of true and predicted values above each threshold shows that the true values of the goods outside her bundle exceed $\rho_i$ times their predicted values by at most $\rho_i-1$ times the predicted value of her bundle in total. For overestimates, her bundle may be worth less than its predicted value, and we use the EEFX partition for the predicted values to bound how much it can lose. For the impossibility results, consistency forces a particular allocation of all goods except the last one, and changing the value of the last good makes this allocation unfair.

For three agents, these algorithms also give each agent an EF1 approximation with half her EEFX factor in \cref{thm:eefx-relative}. Indeed, the proof of \cref{lem:eefx-efk} gives $v_i(A_i)\geq(\alpha/2)D_i(A_j)$ for every $j\neq i$ whenever $A_i$ is $\alpha$-EEFX for agent $i$ and every bundle is nonempty. The algorithms ensure nonempty bundles when $m\geq3$; when $m<3$, their allocations are EF1.

\subsection{Impossibility under Larger Overestimates}

While \cref{thm:eefx-relative}(i) gives a positive approximation of EEFX for underestimates with any error factor, we now show that this is not the case for overestimates, even if the algorithm knows the error factor.

\begin{theorem}\label{thm:eefx-over-impossible}
For two agents whose predictions are overestimates with a common error factor $\rho$, a consistent online algorithm can guarantee a positive approximation of EEFX if and only if $\rho<2$. For every $n\geq3$, no consistent online algorithm guarantees a positive approximation of EEFX when all agents' predictions are overestimates with error factor $3/2$. These impossibility results hold even with identical predictions and identical true valuations in which every good has positive value, and even if the algorithm knows $\rho$.
\end{theorem}

For two agents and $\rho<2$, \cref{thm:eefx-relative}(ii) gives the factor $(2-\rho)/\rho^2$. We prove the impossibility results in \cref{app:eefx-over-impossible}. In both constructions, consistency forces the algorithm to leave an agent without any good until the last good arrives, and the allowed error makes this good much less valuable than predicted. For $n\geq3$, \cref{thm:eefx-relative}(ii) gives a positive approximation whenever $\rho<1+1/(2n-3)$, and determining the exact threshold is left open.

\subsection{Additive Guarantees in Terms of Instantiation Error}\label{sec:eefx-additive}

Finally, we consider predictions that may both underestimate and overestimate the true values. As in \citet[Section~5.1]{NeohPetersTeh2026}, a natural instantiation rule assigns to each arriving good a closest unused predicted value. For any instantiation rule, we define the \emph{instantiation error} of agent $i$ as $\eta_i=\sum_{g\in G}|v_i(g)-\widetilde v_i(g)|$. For the rule above, this error is zero when the predictions are exact, but may be larger than the minimum cost of matching agent $i$'s true values to her predicted values. We give the following EF$k$ guarantee.

\begin{proposition}\label[proposition]{prop:inaccurate}
Let $k\geq1$ and $\alpha\in[0,1]$, and consider an online algorithm that guarantees $\alpha$-EF$k$ given exact frequency predictions. If this algorithm is run on the predicted valuations obtained by any instantiation rule, then the resulting allocation satisfies $v_i(A_i)\geq\alpha D_i^{(k)}(A_j)-\eta_i$ for all distinct agents $i,j$.
\end{proposition}

\begin{proof}
Since each predicted value is used exactly once, the predicted valuations form an instance with exact predictions, so the allocation is $\alpha$-EF$k$ with respect to the predicted valuations. Fix distinct agents $i,j$, and let $X\subseteq A_j$ with $|X|\leq k$ be such that $\widetilde v_i(A_i)\geq\alpha\widetilde v_i(A_j\setminus X)$. Then
\begin{align*}
v_i(A_i)
&\geq\alpha v_i(A_j\setminus X)
-\sum_{g\in A_i}|v_i(g)-\widetilde v_i(g)| -\alpha\sum_{g\in A_j\setminus X}|v_i(g)-\widetilde v_i(g)|\\*
&\geq\alpha D_i^{(k)}(A_j)-\eta_i,
\end{align*}
where the last inequality holds because $A_i$ and $A_j$ are disjoint and $\alpha\leq1$.
\end{proof}

The same argument applies to each comparison in an EEFX partition. Hence, running the algorithm of \cref{thm:eefx} on the predicted valuations gives every agent $i$ a partition of the goods outside her bundle such that $v_i(A_i)\geq v_i(Y\setminus\{g\})-\eta_i$ for every part $Y$ and every $g\in Y$. With the rule above, $\eta_i=0$ when the predictions are exact, so applying \cref{prop:inaccurate} to the algorithms of \cref{thm:three,thm:four} recovers their guarantees.

For each agent, online instantiation is an instance of online minimum-cost matching on the line, where the predicted values are fixed servers and the true values arrive as requests \citep{KalyanasundaramPruhs1993}. Since \cref{prop:inaccurate} holds for any instantiation rule, any online matching algorithm can be used. For example, the deterministic algorithm of \citet{Raghvendra2018} ensures that $\eta_i$ is at most $O(\log m)$ times the minimum matching cost, whereas no online matching algorithm can guarantee a factor of $o(\sqrt{\log m})$ \citep{PesericoScquizzato2023}.

\section{Conclusion}

In this paper, we have studied envy-based fairness in online fair division with frequency predictions. We showed that a positive approximation of EF1 is possible if and only if there are at most four agents, thereby answering the question of \citet{NeohPetersTeh2026}, and that no positive approximation of EFX is possible for three or more agents, even when the valuations take only two distinct positive values. On the positive side, the meta-algorithm of \citet{NeohPetersTeh2026} yields EEFX for any number of agents; on the negative side, algorithms that fix a picking sequence in advance guarantee a positive approximation of EF$k$ only when $n\leq k+2$. For noisy predictions, we obtained consistent algorithms whose first-order losses are the best possible. Overall, our results indicate that frequency predictions suffice for guarantees that depend only on an agent's own bundle, such as shares and EEFX, but not for EF1 when there are five or more agents.

Besides determining the best possible factor for four agents and the exact threshold for overestimates when $n\geq3$, it would be interesting to study frequency predictions for agents with different entitlements \citep{ChakrabortyEtAl2021,SuksompongTeh2022,SuksompongTeh2023,MontanariEtAl2025,Teh2026Weighted}, for groups of agents \citep{KyropoulouSuksompongVoudouris2020,ScarlettTehZick2023}, and for chores or mixed manna \citep{AzizEtAl2022,ElkindIgarashiTeh2024,Teh2026Chores,Teh2026Weighted}, which have also been studied in online and temporal settings \citep{WangWei2026,AmanatidisEtAl2026Buffers,ChoiLiTeh2026}.

\clearpage
\appendix

\bibliographystyle{plainnat}
\bibliography{bib}

\section{Proof of \texorpdfstring{\cref{lem:picking-transfer}}{Lemma 2.2}}\label{app:picking-transfer}

\begin{proof}[Proof of \cref{lem:picking-transfer}]
We give each turn of the picking sequence its original position as a label, which it keeps when other turns are deleted. For a turn of agent $i$ that is currently in position $\ell$ of the remaining sequence, we call agent $i$'s $\ell$th largest remaining predicted value the value \emph{associated} with this turn. We show that the value associated with each turn never decreases until the turn is deleted, and that the good allocated when a turn is deleted is worth at least the value associated with it at that time.

Suppose that a good $g$ arrives and the remaining picking sequence has $s$ turns. For each agent $i$, let $b_{i1}\geq\cdots\geq b_{is}$ be her remaining predicted values, including $v_i(g)$, and let $c_{i1}\geq\cdots\geq c_{i,s-1}$ be these values after removing $v_i(g)$. The simulation runs the remaining picking sequence on $g$ and $s-1$ hypothetical goods, where agent $i$ values the $\ell$th hypothetical good at $c_{i\ell}$. Suppose that the turn in position $r$ picks $g$.

Consider a turn in position $\ell<r$, and let $i$ be its agent. Before $g$ is picked, the simulation picks the hypothetical goods in order of their indices, since every agent values them in this order and ties are broken in favor of smaller indices; hence, this turn picks the $\ell$th hypothetical good. Since ties are broken in favor of $g$, we have $c_{i\ell}>v_i(g)$, and hence $b_{i\ell}=c_{i\ell}$; that is, the value associated with this turn does not change. Next, let $i$ be the agent of the turn in position $r$. If $r<s$, then since this turn picks $g$ rather than the $r$th hypothetical good, we have $v_i(g)\geq c_{ir}$, so at most $r-1$ of her remaining predicted values are larger than $v_i(g)$; hence $v_i(g)\geq b_{ir}$. If $r=s$, then $v_i(g)\geq b_{is}$, since $b_{is}$ is her smallest remaining predicted value. Finally, a turn of agent $i$ in position $\ell>r$ moves to position $\ell-1$, and its associated value changes from $b_{i\ell}$ to $c_{i,\ell-1}\geq b_{i\ell}$.

Initially, the turn labeled $r$ is in position $r$, and the value associated with it is $a_{ir}=v_i^{\IDO}(g_r)$, where $i$ is its agent. Mapping each good in $A_i$ to $g_r$, where $r$ is the label of the turn deleted upon the arrival of this good, therefore gives a bijection $\varphi_i:A_i\to B_i$ with $v_i(g)\geq v_i^{\IDO}(\varphi_i(g))$ for every $g\in A_i$. This proves part~(i).

For part~(ii), fix an agent $i$ and a threshold $t$. By part~(i), the number of goods in $A_i$ that agent $i$ values at least $t$ is at least the number of goods in $B_i$ with value at least $t$ under $v_i^{\IDO}$. Since agent $i$'s true values and her values under $v_i^{\IDO}$ form the same multiset, the number of goods in $G\setminus A_i$ that agent $i$ values at least $t$ is at most the number of goods in $G\setminus B_i$ with value at least $t$ under $v_i^{\IDO}$. If agent $i$'s $r$th largest value for a good in $G\setminus A_i$ were larger than her $r$th largest value under $v_i^{\IDO}$ for a good in $G\setminus B_i$, then taking $t$ to be the former value would give at least $r$ such goods in $G\setminus A_i$ and at most $r-1$ in $G\setminus B_i$, a contradiction.
\end{proof}

\section{Proof of \texorpdfstring{\cref{thm:fixed-efk}}{Theorem 3.5}}\label{app:fixed-efk}

We prove the impossibility part of \cref{thm:fixed-efk}. In our construction, every agent has $k+1$ large values. We first use up all turns of one agent on goods of small value, and then give the goods that she values highly to another agent. The first step is possible because at least two agents have no turn among the first $k+1$ positions.

\begin{proof}
Fix $n\geq k+3$, $\alpha\in(0,1]$, and an algorithm of the stated form. Let $m=nk+1$ and $\ell=k(n-1)$, and choose
\[
H>\max\left\{\ell,\frac{\ell(\ell+1)}{2\alpha}\right\}.
\]
Every agent's frequency multiset consists of $k+1$ copies of $H$ and the values $1,2,\ldots,\ell$. We first describe the construction without normalization. Let $\pi$ be the picking sequence chosen by the algorithm, and let $B_1,\ldots,B_n$ be the IDO bundles as in \cref{sec:transfer}. If some $B_i$ is empty, then $A_i$ is empty by \cref{lem:picking-transfer}, and since $m>nk$, some other bundle contains at least $k+1$ goods; as all values are positive, the allocation is not $\alpha$-EF$k$. Hence, we may assume that every agent has at least one turn.

Since $n\geq k+3$, at least two agents have no turn among the first $k+1$ positions. Among these agents, let $i$ be one with the fewest turns, and let $q=|B_i|$. All turns of these agents lie among the last $\ell$ positions, so $q\leq\lfloor\ell/2\rfloor$. Since $\ell=k(n-1)\geq2k+1$, it follows that $q+k+1\leq\ell$. Moreover, since $m>nk$, some agent has at least $k+1$ turns, and such an agent $j$ can be chosen to be different from $i$. Indeed, otherwise, another agent with no turn among the first $k+1$ positions would have at most $k<|B_i|$ turns, contradicting the choice of $i$.

\paragraph{Allocating the goods of small value.}
For $t=1,\ldots,q$, the good $g_t$ has value $\ell-t+1$ to agent $i$ and value $t$ to every other agent. We claim that all of these goods are allocated to agent $i$. Suppose that this holds for $g_1,\ldots,g_{t-1}$, so only turns of agent $i$ have been deleted. Then her first remaining turn still comes after the first $k+1$ positions. Every other agent values $g_t$ at her unique smallest remaining predicted value, so in the simulation, every turn before agent $i$'s first remaining turn picks a hypothetical good rather than $g_t$, and these hypothetical goods are picked in order of their indices. When agent $i$'s turn comes, at least $k+1$ hypothetical goods have been picked, so none of her hypothetical goods of value $H$ remains. Since $g_t$ has her largest remaining value among the values $1,\ldots,\ell$, she picks $g_t$. This proves the claim, and all $q$ turns of agent $i$ are used.

\paragraph{Allocating the goods of large value.}
For $t=1,\ldots,k+1$, the good $g_{q+t}$ has value $H$ to agents $i$ and $j$, and value $q+t$ to every other agent; these values are available since $q+k+1\leq\ell$. Agent $i$ has no remaining turn. Every agent other than $i$ and $j$ again values the current good at her unique smallest remaining predicted value, so her turns before the first remaining turn of agent $j$ pick hypothetical goods. Agent $j$ picks the current good, since it has her largest value $H$ and ties are broken in its favor. Since agent $j$ has at least $k+1$ turns and received none of the first $q$ goods, all of the goods $g_{q+1},\ldots,g_{q+k+1}$ are allocated to her.

The remaining goods have values given by each agent's remaining predicted values in an arbitrary order. Agent $i$ receives no further good, whereas $A_j$ contains $k+1$ goods that she values at $H$. Therefore,
\[
v_i(A_i)=\sum_{t=1}^{q}(\ell-t+1)\leq\frac{\ell(\ell+1)}2<\alpha H
\]
and $D_i^{(k)}(A_j)\geq H$, so the allocation is not $\alpha$-EF$k$.

To make each agent's values distinct, we replace the $k+1$ copies of $H$ by $H,H+1,\ldots,H+k$, and let $g_{q+t}$ have value $H+k+1-t$ to agents $i$ and $j$. These are their largest remaining values, and all of them are at least $H$, so the argument still applies. To obtain normalized valuations, we divide all predicted values by their sum from the start, consider the picking sequence that the algorithm chooses for these predictions, and divide all true values in the construction by the same sum; this does not affect the argument. Finally, if the algorithm chooses the picking sequence at random, we fix a sequence that is chosen with positive probability and use the instance constructed for it. Conditioned on this sequence, the final allocation is not $\alpha$-EF$k$, so the algorithm does not guarantee $\alpha$-EF$k$.
\end{proof}

\section{Proof of \texorpdfstring{\cref{thm:two-value-efx}}{Theorem 4.2}}\label{app:two-value-efx}

\begin{proof}[Proof of \cref{thm:two-value-efx}]
Let $K>\max\{1,1/\alpha\}$ be an integer, $m=(n-2)(K+1)+2$, and $\ell=1/(2Km)$. Every agent's frequency multiset consists of one value $1$ and $m-1$ copies of $\ell$. Each of the first $m-2$ goods has value $\ell$ to every agent. Let $P$ be the set of agents who receive at least one of these goods. Every agent in $P$ values $g_{m-1}$ at $1$ and $g_m$ at $\ell$, while every agent outside $P$ values $g_{m-1}$ at $\ell$ and $g_m$ at $1$. These values are consistent with the predictions.

If some bundle is empty, then since $m\geq n$, some other bundle contains at least two goods; as all values are positive, the allocation is not $\alpha$-EFX. Hence, we may assume that every bundle is nonempty.

\paragraph{Case 1: $g_{m-1}$ is allocated to some agent $p\in P$.}
Only $g_m$ can be allocated to an agent outside $P$, so as every bundle is nonempty, $|P|\geq n-1\geq2$. Let $i\in P\setminus\{p\}$. Agent $i$ receives no good that she values at $1$, so $v_i(A_i)\leq m\ell=1/(2K)$. On the other hand, $A_p$ contains one of the first $m-2$ goods, and removing it from $A_p$ leaves a bundle containing $g_{m-1}$, which agent $i$ values at $1$. Since $1/(2K)<\alpha$, the allocation is not $\alpha$-EFX.

\paragraph{Case 2: $g_{m-1}$ is allocated to some agent $c\notin P$.}
If $g_m$ is also allocated to $c$, let $i$ be any agent in $P$. Then $v_i(A_i)\leq m\ell$, whereas removing $g_m$ from $A_c$ leaves $g_{m-1}$, which agent $i$ values at $1$. If $g_m$ is allocated to some agent $p\in P$, then $A_c=\{g_{m-1}\}$ and $v_c(A_c)=\ell$, whereas removing one of the first $m-2$ goods from $A_p$ leaves a bundle containing $g_m$, which agent $c$ values at $1$. In both cases, the allocation is not $\alpha$-EFX.

Otherwise, $g_m$ is allocated to an agent outside $P$ other than $c$, so $|P|\leq n-2$. Since the first $m-2$ goods are $(n-2)(K+1)$ goods allocated to the agents in $P$, some agent $j\in P$ receives at least $K+1$ of them. Agent $c$ has $v_c(A_c)=\ell$, while removing any good from $A_j$ leaves value at least $K\ell$ to her. Since $1/K<\alpha$, the allocation is not $\alpha$-EFX.

To obtain normalized valuations, we divide all values by $1+(m-1)\ell$ from the start, which does not affect the comparisons. Since the adversary depends only on the allocation of the first $m-2$ goods, the argument also applies to randomized algorithms.
\end{proof}

\section{Proof of \texorpdfstring{\cref{lem:four-selection}}{Lemma 5.1}}\label{app:four-selection}

We prove \cref{lem:four-selection}. Every good that the selection rule ignores is allocated upon arrival by the potential rule in \cref{sec:four}. For each agent $i$, her first four occurrences have ranks $3,3,2,1$, respectively, and all other goods have rank $0$ for her; we write $\rank_i(g)$ for the rank of good $g$ for agent $i$. Let $u_i=\max_{g\in P_i}\rank_i(g)$ denote the largest rank of a selected good that agent $i$ has received, where $u_i=0$ if $P_i$ is empty. We will ensure that each agent receives at most two selected goods and ends with $u_i>0$. We will also ensure that every good whose rank for some agent $i$ is larger than the final value of $u_i$ is selected, and that each agent other than $i$ receives at most one such good.

\subsection{A Rule for Targets and Marked Goods}

In this subsection, we consider two agents $p$ and $q$, each of whom has two designated goods called her \emph{targets}; an agent is \emph{satisfied} once she receives one of her targets. In addition, at most two goods are \emph{marked}; these goods must be selected and allocated to different agents. In the applications below, they are another agent's two most valuable occurrences, and giving them to different agents allows each to be removed in the corresponding EF1 comparison. The algorithm learns whether a good is a target or marked when the good arrives. Unless specified otherwise, ties are broken in favor of $p$ over $q$.

\begin{lemma}\label[lemma]{lem:three-marked-rule}\label[lemma]{lem:two-marked-rule}
Suppose that there may additionally be a third agent $c$, whose goods have labels in $\{0,1,2\}$, such that either $c$ has no selected good and her remaining goods of positive label have labels $1,2,2$, or $c$ already has one selected good of label $1$ and her remaining goods of positive label have labels $2,2$. Then there exists an online rule that satisfies both $p$ and $q$, allocates the marked goods to different agents, and gives each agent at most two selected goods, including any selected good that she already has. Moreover, if $c$ is present, each selected good that she receives has a larger label than all of her previous selected goods, and her largest label is positive at the end; if this largest label is $1$, then her two goods of label $2$ are allocated to different agents in $\{p,q\}$.
\end{lemma}

\begin{proof}
Let $z$ denote the largest label of a selected good that $c$ has received, where $z=0$ if she has none. If $c$ is absent, let $z=2$ and let all goods have label $0$; the exceptions below then prevent any good from being given to $c$. The rule selects every marked good, every target of an unsatisfied agent, and every good whose label is larger than $z$, and ignores all other goods. For a selected good $x$ that is a target of an unsatisfied agent $r\in\{p,q\}$, we say that $r$ is \emph{eligible} for $x$ unless (a)~$x$ is marked and $r$ has received the other marked good, or (b)~$x$ has label $2$, $z<2$, and $r$ has already received a good of label $2$. If some unsatisfied agent is eligible for $x$, the rule gives $x$ to such an agent. Otherwise, it gives $x$ to $c$, with the following three exceptions.
\begin{enumerate}[label=\textup{(\roman*)},leftmargin=2em]
\item If $x$ is the first marked good and its label is at most $z$, then $x$ is given to $p$ or $q$, preferring a satisfied agent.
\item If $x$ is the second marked good and the first one was given to $c$, then $x$ is given to $p$ or $q$, preferring a satisfied agent; if $x$ has label $2$ and $z<2$, the rule avoids an agent who has already received a good of label $2$.
\item If $x$ is the second marked good, the first one was given to $p$ or $q$, and the label of $x$ is at most $z$, then $x$ is given to the other agent in $\{p,q\}$.
\end{enumerate}
Whenever \textup{(i)} or \textup{(ii)} gives a good to an unsatisfied agent, this agent receives priority in all later tie-breaking. A permitted recipient in \textup{(ii)} exists, because at most one good of label $2$ can have arrived before $x$.

We first show that both $p$ and $q$ are satisfied, ignoring for now the goods given by \textup{(iii)}; their labels do not change $z$. At most one of \textup{(i)} and \textup{(ii)} occurs. If neither of them gives a good to an unsatisfied agent, then each unsatisfied agent is always eligible for her targets and loses at most one target to the other agent, so she receives a target. Suppose that an exception gives a good to an unsatisfied agent $r$, and let $s$ be the other agent. In case \textup{(i)}, both agents still have both of their targets, and only the second marked good can make $r$ ineligible. In case \textup{(ii)}, no marked good arrives later, and an exclusion of type~(b) is possible later only if $x$ is the first good of label $2$ and $z<2$; then $s$ is also unsatisfied, since otherwise she would have been a permitted and preferred recipient. Again, both targets of $r$ remain, and at most one of them can make $r$ ineligible. In every other case, $r$ has a remaining target and cannot become ineligible later. In all cases, her priority ensures that she receives a target. Moreover, $s$ never becomes ineligible while unsatisfied and loses at most one target to $r$, so she also receives a target.

Taking into account the goods given by \textup{(iii)} does not affect this argument. Such a good is the second marked good, so it causes no later exclusion of type~(a); its label is at most $z$, so it causes no exclusion of type~(b); and it is not a target of its recipient if she is unsatisfied, since she would otherwise have been eligible for it. Moreover, it changes neither the priorities nor $z$.

Finally, the rule always allocates the marked goods to different agents, and while $z<2$, it gives the goods of label $2$ to different agents. If $c$ initially has no selected good, she cannot lose all three of her goods of positive label, since before her first selected good, only the two targets given to $p$ and $q$ can take them. Every good given to $c$ increases $z$, so she receives at most two selected goods. Each of $p$ and $q$ receives at most one target and at most one other marked good, and hence at most two selected goods.
\end{proof}

\subsection{The Initial Selections and the Rank-Three Case}

The rule ignores every good that has rank $0$ for all agents. Let $g$ be the first good that has rank at least $2$ for some agent. If $g$ is the first good that has positive rank for some agent, the rule gives $g$ to an agent for whom $g$ has the largest rank. Otherwise, the rule gives the first good with positive rank for some agent to an agent $b$ for whom it has rank $1$, and ignores all subsequent goods before $g$, each of which has rank at most $1$ for every agent. In the latter case, when $g$ arrives, the rule gives $g$ to an agent other than $b$ for whom $g$ has the largest rank among the agents other than $b$, provided that this rank is at least $2$, and gives $g$ to $b$ otherwise. In both cases, let $f$ be the recipient of $g$. If $b$ exists and $b\neq f$, let $c=b$; otherwise, let $c$ be any agent other than $f$. We call the two remaining agents $p$ and $q$.

Thus $u_f\in\{2,3\}$, and $p$ and $q$ have no selected good. Since $g$ is the only good so far that has rank at least $2$ for some agent, each of $p$ and $q$ still has two remaining occurrences of rank at least $2$, at least one of which has rank $3$. If $c$ has no selected good, then at most one of her occurrences of positive rank has arrived, since $\rank_c(g)\leq1$ whenever $b$ exists; hence, she still has at least three remaining occurrences of positive rank. Every good ignored so far has rank at most $1$ for every agent.

Suppose that $u_f=3$. We designate the two largest remaining occurrences of $p$ as her targets. For $c$, we designate her three largest remaining occurrences of positive rank if she has no selected good, and all of her remaining occurrences of positive rank if $c=b$. For each arriving good $x$, the rule follows the first of the following instructions that applies, and ignores $x$ if none applies.
\begin{enumerate}[label=\textup{(\roman*)},leftmargin=2em]
\item If $x$ is the first occurrence of rank $3$ of $q$ to arrive after $g$, give $x$ to $q$.
\item If $x$ is a target of $p$ and $p$ has not yet received a target, give $x$ to $p$.
\item If $|P_c|<2$ and $x$ is a designated occurrence of $c$ with $\rank_c(x)>u_c$, give $x$ to $c$.
\item If $p$ has received a target and $\rank_p(x)>u_p$, give $x$ to $p$; however, if $|P_c|=2$, $\rank_c(x)>u_c$, and $P_p$ already contains a good whose rank for $c$ is larger than $u_c$, give $x$ to $q$ instead.
\item If $|P_c|=2$ and $\rank_c(x)>u_c$, give $x$ to $f$ if $c=b$, $\rank_c(g)=0$, and $|P_f|=1$; otherwise, give $x$ to $q$.
\end{enumerate}
Agent $q$ receives an occurrence of rank $3$ by \textup{(i)}. Agent $p$ loses at most one target to $q$, so she receives a target by \textup{(ii)}, and $u_p\geq2$. If $c$ initially has no selected good, she loses at most two designated occurrences before receiving one. The bounds of two selected goods for $f$, $p$, and $c$ follow directly from the instructions.

Once $c$ has two selected goods, at most one of her occurrences arriving after $g$ has rank larger than $u_c$, except when $c=b$ and $\rank_c(g)=0$, in which case there may be two. In the former case, $q$ receives at most this one additional good. In the latter case, by \textup{(iv)} and \textup{(v)}, $q$ receives an additional good only if the other occurrence of rank larger than $u_c$ is already in $P_p$ or in $P_f$. Hence, $q$ receives at most one additional good, and she cannot receive another occurrence of rank larger than $u_c$ in addition to her own occurrence of rank $3$.

If $p$ ends with $u_p=2$, then one of her occurrences of rank $3$ is $g$, and at most one arrives afterwards; by \textup{(ii)} and \textup{(iv)}, the latter is not given to $f$. For $c$, since targets have priority, at most one occurrence of rank larger than $u_c$ is added to each of $P_p$ and $P_q$ before $c$ receives her second selected good. Afterwards, the second part of \textup{(iv)} prevents $P_p$ from receiving a second such occurrence, the preceding paragraph shows the same for $P_q$, and \textup{(v)} gives such an occurrence to $f$ only if $\rank_c(g)=0$ and $|P_f|=1$.

Finally, every occurrence of rank larger than the final maximum rank of its agent is selected. This is immediate for $f$ and $q$, whose final maximum rank is $3$. For $p$, such an occurrence has rank $3$; if it arrives before she receives a target, it is one of her targets, so it is selected by \textup{(ii)} or \textup{(iv)} unless an earlier instruction applies. For $c$, such an occurrence that arrives after $g$ has rank at least $2$ and is therefore designated, so it is selected by \textup{(iii)} or \textup{(v)} unless an earlier instruction applies. This proves all required properties.

\subsection{The Rank-Two Case}

Suppose that $u_f=2$. Since $g$ has rank at most $2$ for each of $f$, $p$, and $q$, each of them still has both of her occurrences of rank $3$. We use the occurrences of rank $3$ of $p$ and $q$ as their targets.

\paragraph{Agent $c$ has no selected good, or $\rank_c(g)=2$.}
If $c$ has no selected good, her three largest remaining ranks are $1,3,3$ or $2,3,3$; we give the corresponding occurrences labels $1,2,2$, respectively, and give all of her other goods label $0$. If $c=b$ and $\rank_c(g)=2$, then $c$ already has a selected good of rank $1$, which we give label $1$, and her remaining occurrences of positive rank have rank $3$, which we give label $2$. We apply the rule of \cref{lem:three-marked-rule} with the two occurrences of rank $3$ of $f$ as marked goods, and give $f$ no further selected good. The lemma gives $p$ and $q$ an occurrence of rank $3$, and allocates the two occurrences of rank $3$ of $f$ to different agents. If $c$ does not end with $u_c=3$, then her two remaining occurrences of rank $3$ are given to different agents in $\{p,q\}$, and the only earlier good that may have larger rank for her is $g$, which is given to $f$. All required properties follow.

It remains to consider the case where $c=b$ and $\rank_c(g)\in\{0,3\}$. Then $f$ and $c$ each have one selected good.

\paragraph{Agent $c=b$ and $\rank_c(g)=3$.}
The remaining occurrences of positive rank of $c$ have ranks $2$ and $3$. When a good $x$ arrives that is a target of an unsatisfied agent in $\{p,q\}$, the rule gives $x$ to such an agent, with the restriction that while $u_f=2$, an agent who has received an occurrence of rank $3$ of $f$ cannot receive another one. Otherwise, the rule follows the first of the following instructions that applies, and ignores $x$ if none applies.
\begin{enumerate}[label=\textup{(\roman*)},leftmargin=2em]
\item If $|P_c|=1$ and $\rank_c(x)>u_c$, give $x$ to $c$.
\item If $u_f=2$, $\rank_f(x)=3$, and $\rank_c(x)\leq u_c$, give $x$ to $f$.
\item If $\rank_c(x)>u_c$, give $x$ to $p$ or $q$, respecting the same restriction for occurrences of rank $3$ of $f$ and preferring a satisfied agent. If the recipient is unsatisfied, she receives priority afterwards.
\end{enumerate}
Instruction \textup{(iii)} applies at most once, namely to the remaining occurrence of rank $3$ of $c$ after she has received her occurrence of rank $2$. If both $p$ and $q$ are unsatisfied at that time, they still have both of their targets; the recipient can be excluded by the restriction at most once, so her priority ensures that she receives a target. If the other agent is satisfied, she can have been passed over only because the current good is the second occurrence of rank $3$ of $f$, so no further exclusion can occur. The other agent in $\{p,q\}$ is never excluded while unsatisfied and loses at most one target. Hence, both $p$ and $q$ receive an occurrence of rank $3$ and at most two selected goods.

While $u_f=2$, the two occurrences of rank $3$ of $f$ are selected and given to different agents: the restriction and \textup{(iii)} ensure this for $p$ and $q$, and $c$ receives only one additional good. For $c$, every occurrence of rank larger than $u_c$ is selected, and by \textup{(ii)}, $f$ receives no such occurrence other than $g$. If $c$ ends with $u_c=1$, then her remaining occurrences of ranks $2$ and $3$ were given to different agents in $\{p,q\}$; if $c$ ends with $u_c=2$, then her remaining occurrence of rank $3$ is not given to $f$. The bounds on the number of selected goods and the required properties follow immediately.

\paragraph{Agent $c=b$ and $\rank_c(g)=0$.}
Let $y$ be the remaining occurrence of rank $2$ of $c$. Before $y$ arrives, the rule gives each arriving good that has rank $3$ for some agent whose current maximum rank is less than $3$ to such an agent, with priority order $p,q,c,f$, and ignores all other goods. When $y$ arrives, the rule proceeds as follows.
\begin{enumerate}[label=\textup{(\roman*)},leftmargin=2em]
\item If $u_c=3$ or $y$ is a target of an unsatisfied agent, the rule continues with the same priority rule. Both $p$ and $q$ then receive an occurrence of rank $3$. If necessary, so does $c$: one of the targets given to $p$ and $q$ is $y$, so only one occurrence of rank $3$ of $c$ can be given to $p$ or $q$. If $f$ ends with $u_f=2$, her two occurrences of rank $3$ are given to different agents, since every other agent receives at most one further selected good.
\item Otherwise, if $u_f=3$, the rule gives $y$ to $c$ and continues with the rule of \cref{lem:two-marked-rule} for $p$ and $q$, with the two occurrences of rank $3$ of $c$ as marked goods.
\item Otherwise, the rule gives $y$ to $f$ and continues with the rule of \cref{lem:three-marked-rule} for $p$, $q$, and $c$, where the selected good of $c$ has label $1$ and her occurrences of rank $3$ have label $2$. If $u_f$ is still $2$, the two occurrences of rank $3$ of $f$ are marked; otherwise, no good is marked.
\end{enumerate}
In \textup{(ii)} and \textup{(iii)}, we may regard the rule of \cref{lem:three-marked-rule} as having been applied since $g$ was allocated. Indeed, since $u_c<3$ when $y$ arrives, and in \textup{(iii)} also $u_f<3$, every good selected after $g$ and before $y$ was given to an unsatisfied agent in $\{p,q\}$ as one of her targets, except for the occurrence of rank $3$ selected for $f$ in \textup{(ii)}. Such an agent had no selected good and was therefore eligible, so the rule of \cref{lem:three-marked-rule} would have made the same decision. Every other good that arrived after $g$ and before $y$, as well as $y$ itself, is neither marked nor a target of an unsatisfied agent and has label $0$, so this rule would have ignored it; giving such a good to $f$ or $c$ does not affect the guarantees of \cref{lem:three-marked-rule}.

By \cref{lem:three-marked-rule}, every agent receives at most two selected goods, and the required occurrences of larger rank are given to different agents. In \textup{(iii)}, if $c$ ends with $u_c=1$, her occurrences of larger rank are $y$, which is given to $f$, and her two occurrences of rank $3$, which are given to different agents in $\{p,q\}$; otherwise, $c$ ends with $u_c=3$. Every occurrence of rank larger than the final maximum rank of its agent is selected.

\paragraph{Returning to values.}
In every case, each agent receives one or two selected goods and ends with $u_i>0$, so she receives one of her first four occurrences. Every good ignored before $g$ has rank at most $1$ for every agent, and every good ignored afterwards has rank at most the final value of $u_i$ for every agent $i$. Since the two largest values of each agent are equal after the first capping step, every good of value larger than $h_i$ to agent $i$ has rank larger than $u_i$ for her. These goods are selected and given to different agents other than $i$. Hence, for every agent $j\neq i$, at most one good in $A_j$ has value larger than $h_i$ to agent $i$, regardless of how the goods that are not selected are allocated. This proves \cref{lem:four-selection}.

\section{Proof of \texorpdfstring{\cref{lem:four-potentials}}{Lemma 5.2}}\label{app:four-potentials}

We prove \cref{lem:four-potentials}. Let $i$ be an agent with $a_{i4}>0$, and define the nonincreasing step function $p_i$ on $[0,\infty)$ as follows. For $0\le s<x_i+a_{i4}$, let
\[
p_i(s)=\frac1{(C-1)a_{i4}+Cx_i},
\]
and for every integer $k\ge1$ and every $s$ with $x_i+ka_{i4}\le s<x_i+(k+1)a_{i4}$, let
\[
p_i(s)=\frac1{(C-1)a_{i4}+C(x_i+ka_{i4})}.
\]
Let $\Phi_i(u)=\int_0^{\min\{u,\,T_i+a_{i4}\}}p_i(s)\,ds$ for $u\geq0$. This function is nondecreasing, concave, and piecewise linear, with $\Phi_i(0)=0$. If $u<x_i$, its increase over $[u,u+a_{i4}]$ is $a_{i4}/((C-1)a_{i4}+Cx_i)$. Otherwise, write $u=x_i+(k+\theta)a_{i4}$, where $k\ge0$ is an integer and $\theta\in[0,1)$. For $u\le T_i$, the convexity of the reciprocal function gives
\begin{align*}
 \Phi_i(u+a_{i4})-\Phi_i(u)
 &=\frac{(1-\theta)a_{i4}}{(C-1)a_{i4}+C(x_i+ka_{i4})} +\frac{\theta a_{i4}}{(C-1)a_{i4}+C(x_i+(k+1)a_{i4})}\\
 &\ge\frac{a_{i4}}{(C-1)a_{i4}+Cu}.
\end{align*}
This proves the required lower bound on the increase.

To show that $\Phi_i\le1$, we normalize $a_{i4}=1$, omit the subscript $i$, and write $q=3+1/C=6-C$ and $d=C-1+Cx$. Let
\begin{gather*}
 r=\lfloor2+(q-1)x\rfloor,\quad
 \theta=2+(q-1)x-r,\\
 E=d+C(r+1)=qd+4-C\theta.
\end{gather*}
Since $T=2+qx$, we have $T+1=(x+1)+r+\theta$, so the interval $[0,T+1)$ consists of $[0,x+1)$, the intervals $[x+k,x+k+1)$ for $k\in[r]$, and a final interval of length $\theta$ on which the slope is $1/E$. Hence, the total increase is
\[
 S(x):=\Phi(T+1)
 =\frac{x+1}{d}+\sum_{k=1}^{r}\frac1{d+Ck}+\frac\theta E.
\]
We show that $S$ is decreasing in $x$. Away from its breakpoints,
\[
 S'(x)=-\frac1{d^2}-C\sum_{k=1}^{r}\frac1{(d+Ck)^2}
          +\frac{q-1}{E}-\frac{C\theta}{E^2}.
\]
The trapezoid bound for the convex function $t\mapsto1/(d+Ct)^2$ yields
\[
 C\sum_{k=1}^{r}\frac1{(d+Ck)^2}
 \ge\frac1{d+C}-\frac1E
       +\frac C2\left(\frac1{(d+C)^2}-\frac1{E^2}\right).
\]
For fixed $d$, the function $\theta\mapsto q/E+C(1/2-\theta)/E^2$ is increasing on $[0,1]$: its derivative is $C((q-1)E+C(1-2\theta))/E^3>0$, since $E\ge3C-1>C$ and $q-1>1$. Its value at $\theta=1$ is less than $q/(qd+4-C)$. Hence, writing $R=(4C-3)/q$, we obtain
\begin{align*}
 S'(x)
 &\le-\frac1{d^2}-\frac1{d+C}-\frac C{2(d+C)^2}
                  +\frac q{qd+4-C}\\
 &\le-\frac1{d^2}-\frac C{2(d+C)^2}
                  +\frac R{d(d+C)}<0.
\end{align*}
The final inequality follows from the inequality of arithmetic and geometric means and $R^2<2C$; indeed, using $C^2=3C-1$, we have $R=(21C-13)/19$ and $R^2-2C=(55C-272)/361<0$. Since $S$ is continuous at its breakpoints, it follows that
\[
 S(x)\le S(0)=\frac1{C-1}+\frac1{2C-1}+\frac1{3C-1}=1.
\]
This proves \cref{lem:four-potentials}. The construction uses $O(m)$ intervals of constant slope, because $F_i\le\max\{m-7,0\}\,a_{i4}$. Computing each potential increase by summing over these intervals takes $O(m)$ time, giving an $O(m^2)$-time implementation of the four-agent algorithm.

\section{Proof of \texorpdfstring{\cref{thm:four}}{Theorem 5.3}}\label{app:four-proof}

\begin{proof}[Proof of \cref{thm:four}]
By the discussion before \cref{lem:four-potentials} in \cref{sec:four}, it suffices to show that $D_i^w(A_j)\leq Cw_i(A_i)$ for all distinct agents $i,j$. Fix such agents, and let $\ell_i$ denote its final value. If $a_{i4}=0$, then $w_i\equiv0$ and the claim is immediate, so assume that $a_{i4}>0$.

\paragraph{Unselected goods.}
Consider a good $g\in A_j\setminus P_j$. When $g$ arrives, the increase in agent $j$'s potential is at least the increase in agent $i$'s potential that $g$ would cause. Since $\ell_i$ never decreases and $\Phi_i$ is concave, the latter increase is at least $\Phi_i(\ell_i+w_i(g))-\Phi_i(\ell_i)$, evaluated at the final value of $\ell_i$. As agent $j$'s potential never decreases and is at most $1$, summing over these goods gives
\[
\sum_{g\in A_j\setminus P_j}[\Phi_i(\ell_i+w_i(g))-\Phi_i(\ell_i)]\leq1.
\]
Since $w_i(g)\leq a_{i4}$, concavity implies that each term is at least $(w_i(g)/a_{i4})[\Phi_i(\ell_i+a_{i4})-\Phi_i(\ell_i)]$. By \cref{lem:four-potentials}, if $\ell_i\leq T_i$, then
\[
w_i(A_j\setminus P_j)\leq(C-1)a_{i4}+C\max\{x_i,\ell_i\}.
\]
We refer to this as the bound on unselected goods; it holds for every agent $j\neq i$.

\paragraph{A lower bound on $\ell_i$.}
Remove from $G$ the selected goods of the other three agents and one good in $P_i$ among agent $i$'s first four occurrences. At most seven goods are removed, and agent $i$'s seven largest values under $w_i$ are those of her first seven occurrences, so the remaining goods have total value at least $F_i$ under $w_i$. These goods consist of agent $i$'s remaining goods, whose value under $w_i$ is $\ell_i$, and the unselected goods of the other three agents. If $\ell_i<x_i$, then $x_i>0$ and $\ell_i<T_i$, so applying the bound on unselected goods to each of the other three agents yields
\[
F_i<x_i+3[(C-1)a_{i4}+Cx_i]=3(C-1)a_{i4}+(3C+1)x_i=F_i,
\]
a contradiction. Hence $\ell_i\geq x_i$.

\paragraph{The case $\ell_i\leq T_i$.}
Since $|P_j|\leq2$, removing one selected good from $A_j$ leaves at most one selected good, whose value under $w_i$ is at most $a_{i4}$. By the bound on unselected goods and $\ell_i\geq x_i$, we obtain
\[
D_i^w(A_j)\leq a_{i4}+w_i(A_j\setminus P_j)\leq Ca_{i4}+C\ell_i=Cw_i(A_i).
\]

\paragraph{The case $\ell_i\geq T_i$.}
Let $W_i=\sum_{k=5}^m a_{ik}$ be agent $i$'s total value for the goods other than her first four occurrences. Since $a_{ik}\leq a_{i4}$ for every $k\geq5$, we have
\[
W_i\leq3a_{i4}+F_i\leq3Ca_{i4}+(3C+1)x_i=C(a_{i4}+T_i).
\]
If $A_j$ contains none of agent $i$'s first four occurrences, then $D_i^w(A_j)\leq W_i\leq C(a_{i4}+T_i)\leq C(a_{i4}+\ell_i)$. Otherwise, remove one such occurrence from $A_j$, and remove from $A_i$ a good in $P_i$ among her first four occurrences. The remaining goods of $A_i$ and $A_j$ contain at most two of agent $i$'s first four occurrences, so $D_i^w(A_j)+\ell_i\leq2a_{i4}+W_i$. Since $2a_{i4}\leq T_i\leq\ell_i$, this gives $D_i^w(A_j)\leq W_i\leq C(a_{i4}+\ell_i)$.

In both cases, $D_i^w(A_j)\leq C(a_{i4}+\ell_i)=Cw_i(A_i)$, as desired.
\end{proof}

\section{Omitted Proofs in \texorpdfstring{\cref{sec:inaccurate}}{Section 6}}\label{app:eefx-inaccurate}

We prove \cref{thm:eefx-relative,thm:eefx-over-impossible}. Both instantiation rules in \cref{sec:eefx-relative} and the algorithm of \cref{thm:eefx} are deterministic and run in polynomial time. Since every predicted value is used exactly once, the algorithm of \cref{thm:eefx} returns an allocation that is EEFX with respect to the predicted valuations, regardless of the accuracy of the predictions. When the predictions are exact, the true value of each arriving good is among the unused predicted values, so both rules assign to each good its true value; hence, both algorithms are consistent. Moreover, for every agent $i$, there exists a partition $(Z_j)_{j\neq i}$ of $G\setminus A_i$ with $|Z_j|=|A_j|$ for every $j\neq i$ such that, writing $V=\widetilde v_i(A_i)$,
\[
\widetilde v_i(Z_j\setminus\{g\})\leq V\quad\text{for all }j\neq i\text{ and }g\in Z_j.
\]
For agent $i$, we use this partition to prove the claimed guarantees for the true values. This also shows that the guarantee for agent $i$ does not depend on the predictions for the other agents.

\subsection{Proof of \texorpdfstring{\cref{thm:eefx-relative}}{Theorem 6.1}(i)}\label{app:eefx-under}

\begin{proof}[Proof of \cref{thm:eefx-relative}(i)]
Fix an agent $i$ whose predictions are underestimates with error factor $\rho=\rho_i$. We first show that $v_i(g)\geq\widetilde v_i(g)$ for every good $g$. Initially, matching agent $i$'s true values to her predicted values in nonincreasing order gives a matching in which every predicted value is at most the true value matched to it. Suppose that such a matching exists between the true values of the goods that have not been allocated and the unused predicted values when a good $g$ arrives, and let $p$ be the predicted value matched to $v_i(g)$. The rule assigns to $g$ the largest unused predicted value $p^*$ that is at most $v_i(g)$, so $p^*\geq p$. If $p^*$ and $p$ are different occurrences, let $y$ be the true value matched to $p^*$; replacing the pairs $(v_i(g),p)$ and $(y,p^*)$ by $(v_i(g),p^*)$ and $(y,p)$ preserves the property, since $y\geq p^*\geq p$. Removing the pair of $g$ yields the desired matching for the next arrival. Hence, $v_i(g)\geq\widetilde v_i(g)$ for every $g\in G$, and in particular $v_i(A_i)\geq V$.

Let $\psi$ be the bijection from $G\setminus A_i$ to itself that maps the goods, in nonincreasing order of their true values, to the goods in nonincreasing order of their predicted values. Writing $y_+=\max\{y,0\}$, we show that
\[
\sum_{g\in G\setminus A_i}(v_i(g)-\rho\,\widetilde v_i(\psi(g)))_+\leq(\rho-1)V.
\]
For a finite list $y$ of nonnegative numbers, let $N_y(t)$ be the number of entries of $y$ larger than $t$. For two lists $y$ and $z$ of the same length, both in nonincreasing order, we have
\[
\sum_r(y_r-z_r)_+=\int_0^\infty(N_y(t)-N_z(t))_+\,dt,
\]
since for every $t$, the number of indices $r$ with $y_r>t\geq z_r$ is equal to $(N_y(t)-N_z(t))_+$. Since the predictions are underestimates with error factor $\rho$, for every $t$, the number of goods with true value larger than $t$ is at most the number of goods $g$ with $\rho\,\widetilde v_i(g)>t$. Hence, for every $t$, the number of goods in $G\setminus A_i$ with true value larger than $t$ exceeds the number of goods $g\in G\setminus A_i$ with $\rho\,\widetilde v_i(g)>t$ by at most the amount by which the number of goods $g\in A_i$ with $\rho\,\widetilde v_i(g)>t$ exceeds the number of goods in $A_i$ with true value larger than $t$. Applying the identity to both sides, the left-hand side of the desired inequality is at most $\sum_r(y_r-z_r)_+$, where $y$ and $z$ list the values $\rho\,\widetilde v_i(g)$ and $v_i(g)$ for $g\in A_i$, respectively, in nonincreasing order. Since $v_i(g)\geq\widetilde v_i(g)$ for every $g$, each $z_r$ is at least the $r$th largest predicted value of a good in $A_i$, so this sum is at most $(\rho-1)\widetilde v_i(A_i)=(\rho-1)V$, as claimed.

For each $j\neq i$, let $Y_j=\psi^{-1}(Z_j)$, so that $|Y_j|=|Z_j|=|A_j|$. For every $g\in Y_j$,
\begin{align*}
v_i(Y_j\setminus\{g\})
&\le \rho\,\widetilde v_i(Z_j\setminus\{\psi(g)\})+(\rho-1)V\\
&\le (2\rho-1)V
\le (2\rho-1)\,v_i(A_i).
\end{align*}
Hence, $A_i$ is $\frac1{2\rho-1}$-EEFX for agent $i$.
\end{proof}

\subsection{Proof of \texorpdfstring{\cref{thm:eefx-relative}}{Theorem 6.1}(ii): General Bound}\label{app:eefx-over-general}

We prove the first bound in \cref{thm:eefx-relative}(ii).

\begin{proof}
Fix an agent $i$ whose predictions are overestimates with error factor $\rho=\rho_i$. The exchange argument in \cref{app:eefx-under}, with the inequalities reversed, shows that $v_i(g)\leq\widetilde v_i(g)$ for every good $g$: if the arriving value $v_i(g)$ is matched to $p$, the rule assigns the smallest unused predicted value $p^*$ that is at least $v_i(g)$, and if $y$ is matched to $p^*$, then $v_i(g)\leq p^*$ and $y\leq p^*\leq p$. Moreover, in nonincreasing order, each true value of agent $i$ is at least $1/\rho$ times the corresponding predicted value; we refer to these inequalities as the \emph{sorted lower bounds}.

Let $s=v_i(A_i)$. If $V=0$, then every part $Z_j$ with at least two goods consists of goods with predicted value zero, and hence true value zero, so all comparisons hold exactly. Assume therefore that $V>0$. We increase the true value of every good outside $A_i$ to its predicted value. This preserves the sorted lower bounds and can only make the desired comparisons harder. Every good outside $A_i$ with predicted value larger than $V$ forms a part by itself. We temporarily remove these goods from both sorted lists. Since the true and predicted values of the goods in $A_i$ are at most $V$, the removed goods are the largest entries of both lists, so the sorted lower bounds remain valid.

Let $q\leq n-1$ be the number of nonempty parts that remain. If $q=0$, every part contains at most one good, and the comparisons hold. Otherwise, let $W$ be the total value of the remaining goods outside $A_i$, and let $b$ be the smallest of these values. Each remaining part has value at most $V$ plus its smallest value; one of these smallest values is $b$, and each of the others is at most $V$. Hence $W\leq(2q-1)V+b$. Summing the sorted lower bounds over the remaining lists gives $\rho(s+W)\geq V+W$, and therefore
\[
\rho s\geq[1-(2q-1)(\rho-1)]V-(\rho-1)b.
\]
If every good in $A_i$ has predicted value less than $b$, then these goods occupy the last positions in both lists, and the sorted lower bounds give $s\geq V/\rho$, which implies the claimed bound. Otherwise, the number of predicted values at least $b$ exceeds the number of remaining goods outside $A_i$, so by the sorted lower bounds, some good in $A_i$ has true value at least $b/\rho$. Thus $\rho s\geq b$, and
\begin{align*}
\rho s &\geq\max\{b,\ [1-(2q-1)(\rho-1)]V-(\rho-1)b\} \geq\frac{1-(2q-1)(\rho-1)}{\rho}\,V,
\end{align*}
where the second inequality holds because if $b$ is smaller than the right-hand side, then the second term in the maximum is at least the right-hand side. Since $q\leq n-1$, we obtain
\[
s\geq\frac{1-(2q-1)(\rho-1)}{\rho^2}\,V
\geq\frac{2(n-1)-(2n-3)\rho}{\rho^2}\,V.
\]
If the last factor is nonpositive, the claim holds trivially. Otherwise, every remaining part has value at most $V$ after removing any good, which gives the required comparisons. Restoring the removed parts, each of which consists of a single good, and decreasing the values of the goods outside $A_i$ back to their true values preserve these comparisons.
\end{proof}

\subsection{Proof of \texorpdfstring{\cref{thm:eefx-relative}}{Theorem 6.1}(ii): Refined Bound for Three or More Agents}\label{app:eefx-over-refined}

We prove the second bound in \cref{thm:eefx-relative}(ii). Compared with the first bound, we additionally consider the largest predicted value $h$ of a good in $A_i$. If $h$ is small, the parts have smaller total value. If $h$ is large, this good keeps most of its value, and we treat separately the parts consisting of two goods of large value.

\begin{proof}
Let $n\geq3$, and consider the algorithm and the agent $i$ with error factor $\rho=\rho_i$ in \cref{app:eefx-over-general}. By the case $V=0$ handled there, we may assume that $V>0$, and we may also assume that the claimed factor is positive. We scale agent $i$'s values so that $V=1$, and let $s=v_i(A_i)$ and $\varepsilon=\rho-1$. As before, we increase the true value of each good outside $A_i$ to its predicted value and use the sorted lower bounds. Throughout the proof, we consider all $n-1$ parts $Z_j$, including empty parts and parts that become empty after goods are temporarily removed.

We first remove the goods outside $A_i$ with value larger than $1$. Each such good forms a part by itself and is among the largest entries of both lists. Each remaining part has value at most $1$ plus its smallest value, which is at most $1$, so the remaining goods outside $A_i$ have total value at most $2(n-1)$. Summing the sorted lower bounds gives
\[
s\geq\frac{1-2(n-1)\varepsilon}{1+\varepsilon}\geq1-(2n-1)\varepsilon.
\]
Let $h$ be the largest predicted value of a good in $A_i$, and let $d$ be the true value of this good. The other goods in $A_i$ have total predicted value $1-h$, and their true values are at most their predicted values. Hence, $d\geq s-(1-h)\geq h-(2n-1)\varepsilon$. In each of the following cases, we start again from the original lists with the increased values outside $A_i$.

Suppose first that $h\leq1/2+(2n-1)\varepsilon$. We remove all goods outside $A_i$ with value larger than $h$; since all values of goods in $A_i$ are at most $h$, these goods are the largest entries in both lists. Each remaining part has value at most $1+h$. Indeed, if a good was removed from the part, the remaining value is at most $1$; otherwise, the value of the part is at most $1$ plus its smallest value, which is at most $h$. Hence, the total value $W$ of the remaining goods outside $A_i$ is at most $(n-1)(1+h)\leq(n-1)(\tfrac32+(2n-1)\varepsilon)$, and summing the sorted lower bounds gives
\begin{align*}
s&\geq\frac{1-\varepsilon W}{1+\varepsilon}
\geq\frac{1-\tfrac32(n-1)\varepsilon-(n-1)(2n-1)\varepsilon^2}{1+\varepsilon} \geq1-2(n-1)\varepsilon-(n-1)(2n-1)\varepsilon^2,
\end{align*}
where the last inequality uses $n-1\geq2$. Since removing any good from any part leaves value at most $1$, the required comparisons hold.

Now suppose that $h>1/2+(2n-1)\varepsilon$. We remove the goods outside $A_i$ with value larger than $h$, as well as those with value strictly between $1-h$ and $h-(2n-1)\varepsilon$. The sorted lower bounds remain valid. Indeed, the goods in the first group are the largest entries of both lists. For the second group, every good in $A_i$ other than the one of largest predicted value has predicted and true values at most $1-h$, while the latter good has predicted value $h$ and true value $d\geq h-(2n-1)\varepsilon$. Hence, each removed good in the second group has the same position in both lists, after breaking ties consistently, and removing entries in the same positions from both lists preserves the sorted lower bounds for the remaining entries.

Every remaining good outside $A_i$ has value at most $1-h$ or at least $h-(2n-1)\varepsilon>1/2$. If a part contains two goods of the latter kind, then the part consists of exactly these two goods, since otherwise removing a third good, even one of value zero, would leave value larger than $1$; we call such a part a \emph{pair}. A pair has value at most $2h$. Every other part has remaining value at most $2-h$. Indeed, a part from which a good was removed has remaining value at most $1$; a part containing a good of value at most $1-h$ has value at most $1+(1-h)$; and a part with a single good has value at most $h\leq2-h$.

If at least one part is not a pair, then since the number of pairs is between $0$ and $n-2$ and the resulting bound is linear in this number, the total value of the remaining goods outside $A_i$ satisfies
\[
W\leq\max\{(n-1)(2-h),\ 2(n-2)h+2-h\}\leq2n-3,
\]
where the second inequality uses $1/2<h\leq1$ and $n\geq3$. Summing the sorted lower bounds now gives
\[
s\geq\frac{1-(2n-3)\varepsilon}{1+\varepsilon}\geq1-2(n-1)\varepsilon,
\]
which suffices for all parts.

It remains to consider the case in which every part is a pair. Then there are exactly $2n-2$ goods outside $A_i$, none of which was removed. Let $z_1\geq\cdots\geq z_{2n-2}$ be their values, each of which is larger than $1/2$, and let $H=z_1\leq h$. If $d\geq H$, then $s\geq H$, whereas removing either good from a pair leaves value at most $H$, so the comparisons hold exactly. Otherwise, let $r$ be the number of values $z_\ell$ larger than $d$, so that $1\leq r\leq2n-2$. The other goods in $A_i$ have predicted and true values at most $1-h<1/2$, while $d>1/2$. Hence, the list of predicted values begins with $h,z_1,\ldots,z_r$, and the list of true values begins with $z_1,\ldots,z_r,d$, with ties ordered arbitrarily. The sorted lower bounds give $z_\ell\geq z_{\ell-1}/\rho$ for $2\leq\ell\leq r$ and $d\geq z_r/\rho$, so
\[
s\geq d\geq H/\rho^r\geq H/\rho^{2n-2}.
\]
Since removing either good from a pair leaves value at most $H$, the factor for agent $i$ is at least $\rho^{-(2n-2)}\geq1-2(n-1)\varepsilon$, where the last inequality follows from Bernoulli's inequality.

In every case, the partition $(Z_j)_{j\neq i}$ gives the required comparisons for the increased values, and decreasing the values of the goods outside $A_i$ back to their true values preserves them.
\end{proof}

\subsection{Proof of \texorpdfstring{\cref{thm:eefx-relative}}{Theorem 6.1}: Lower Bound for Underestimates}\label{app:eefx-under-local}

We prove the impossibility result for underestimates in \cref{thm:eefx-relative}.

\begin{proof}
Let $\rho>1$ and $a=1/\rho$, and suppose for contradiction that some consistent algorithm guarantees $\beta$-EEFX for some $\beta>1/(\rho+1-a)$ whenever all agents' predictions are underestimates with error factor $\rho$. Choose a positive integer $r$ large enough that $b=(1-a)/r$ satisfies $b<a$ and $1/(\rho+1-a-b)<\beta$, and let $J=\rho+3$. Every agent's predicted multiset consists of $n-2$ copies of $J$, the values $1$ and $a$, and $r$ copies of $b$. The goods arrive with common values: first $n-2$ goods of value $J$, then one good of value $1$, and then $r$ goods of value $b$. If the predictions were exact, the last good would have common value $a$.

Consider an EEFX allocation of the instance with this exact continuation. We claim that the $n-1$ goods of value $J$ or $1$ are allocated to different agents, and that the remaining agent receives all other goods. Consider an agent who receives none of these $n-1$ goods. Her bundle has value at most $a+rb=1$. In any partition showing that her bundle is EEFX for her, each of the $n-2$ goods of value $J$ forms a part by itself. If she does not receive all of the remaining goods, then the last part contains the good of value $1$ and another good of positive value, while her bundle has value less than $1$; removing the latter good contradicts EEFX. Hence, she receives all goods of value $a$ or $b$. Since there can be only one such agent, the claim follows.

Let $i$ be the agent who receives the good of value $1$, and let $j$ be the agent who receives the goods of value $b$. By consistency, all goods of value $b$ are allocated to $j$ before the last good arrives. Now let the last good have common value $\rho$ instead of $a$. The true values other than $J$ are $\rho,1,b,\ldots,b$ in nonincreasing order, and each of them is between the corresponding predicted value among $1,a,b,\ldots,b$ and $\rho$ times this value, so the predictions are underestimates with error factor $\rho$.

Consider an agent in $\{i,j\}$ who does not receive the last good. Her bundle has value at most $1$, and $\beta J>J/(\rho+1-a)>1$, so in any partition showing that her bundle is $\beta$-EEFX for her, each of the $n-2$ goods of value $J$ forms a part by itself, and the last part contains all other goods outside her bundle. If this agent is $j$, her value is $rb=1-a$, while the last part contains the goods of values $1$ and $\rho$, so her factor is at most $(1-a)/\rho$. If this agent is $i$, her value is $1$, while the last part contains the good of value $\rho$ and all goods of value $b$, so removing a good of value $b$ leaves value $\rho+1-a-b$, and her factor is at most $1/(\rho+1-a-b)$. Since $(1-a)^2\leq1=\rho a$, we have $(1-a)(\rho+1-a)\leq\rho$, so $(1-a)/\rho\leq1/(\rho+1-a-b)$. In both cases, the factor is less than $\beta$, a contradiction.
\end{proof}

\subsection{Proof of \texorpdfstring{\cref{thm:eefx-relative}}{Theorem 6.1}: Lower Bound for Overestimates}\label{app:eefx-over-local-lower}

We prove the impossibility result for overestimates in \cref{thm:eefx-relative}.

\begin{proof}
Let $1<\rho<2^{1/(2n-2)}$. Every agent's predicted multiset is $\{1,\rho^{-1},\ldots,\rho^{-(2n-2)}\}$; since $\rho^{2n-2}<2$, all of these values are larger than $1/2$. The goods arrive with common values: first $2n-2$ goods of values $\rho^{-1},\ldots,\rho^{-(2n-2)}$, and then one last good. If the predictions were exact, the last good would have common value $1$.

Consider an EEFX allocation of the instance with this exact continuation. Every agent receives at least one good: otherwise, an agent with an empty bundle has $2n-1$ goods of positive value outside her bundle, so any partition of these goods into $n-1$ parts has a part with at least two goods, and removing one of them leaves positive value. Since there are $2n-1$ goods, some agent receives exactly one good. Suppose that this good has value less than $1$. In any partition showing that her bundle is EEFX for her, the good of value $1$ forms a part by itself, since otherwise removing another good from its part would leave value at least $1$. If $n=2$, no part remains for the other good outside her bundle. If $n\geq3$, the remaining $2n-3$ goods outside her bundle lie in at most $n-2$ parts, so some part contains at least three goods, and removing one of them leaves value larger than $1$. Both cases contradict EEFX. Hence, any agent who receives exactly one good receives the good of value $1$, so exactly one agent receives one good and every other agent receives two goods. In particular, the agent who receives the last good has no good before it arrives. By consistency, the algorithm must leave some agent without any good after the first $2n-2$ goods have been allocated.

Now let the last good have common value $u=\rho^{-(2n-1)}$ instead. The true values are then exactly the predicted values divided by $\rho$, so the predictions are overestimates with error factor $\rho$. By the argument above for empty bundles, the agent $i$ without any good must receive the last good for the allocation to be $\beta$-EEFX for any $\beta>0$. Her value is then $u$, and the goods outside her bundle have values $\rho^{-1},\ldots,\rho^{-(2n-2)}$. For any partition of these $2n-2$ goods into $n-1$ parts, some part has value at least $\rho^{-1}$ after removing some good from it. Indeed, if the part containing the largest good contains another good, we remove the latter good. Otherwise, $n\geq3$, the remaining $2n-3$ goods lie in $n-2$ parts, and some part contains at least three goods, so removing one of them leaves value larger than $1$. Hence, the bundle of agent $i$ is not $\beta$-EEFX for her for any $\beta>u/\rho^{-1}=\rho^{-2(n-1)}$.
\end{proof}

\subsection{Proof of \texorpdfstring{\cref{thm:eefx-over-impossible}}{Theorem 6.2}}\label{app:eefx-over-impossible}

We prove \cref{thm:eefx-over-impossible} using instances with identical predictions and identical true valuations in which every good has positive value. In each instance, consistency forces an agent to receive no good before the last arrival, and we then make the last good much less valuable than predicted.

\begin{proof}[Proof of \cref{thm:eefx-over-impossible}]
For two agents and $\rho<2$, \cref{thm:eefx-relative}(ii) gives the positive factor $(2-\rho)/\rho^2$. Since EEFX is equivalent to EFX for two agents, it remains to show that no consistent algorithm guarantees a positive approximation of EFX when the predictions are overestimates with error factor $2$; this also covers every larger error factor.

Suppose for contradiction that some consistent algorithm guarantees $\beta$-EFX for some $\beta>0$, and let $h\geq2$ be an integer such that $1/(2(2^h-2))<\beta$. Both agents' predicted multisets are $\{2^0,2^1,\ldots,2^h\}$, and the first $h$ goods arrive with common values $2^0,\ldots,2^{h-1}$, in this order. If the predictions were exact, the last good would have common value $2^h$, which is larger than the total value $2^h-1$ of the first $h$ goods. If both agents received a good before the last arrival, then the agent who receives the last good would also have an earlier good of positive value; removing this earlier good would leave her bundle with value at least $2^h$, which is larger than the value of the other agent's bundle. Hence, consistency forces the algorithm to allocate all of the first $h$ goods to one agent.

Now let the last good have common value $1/2$ instead. The true values are then exactly half of the predicted values, so the predictions are overestimates with error factor $2$. The agent without any good must receive the last good for the allocation to be $\beta$-EFX. Her bundle then has value $1/2$, whereas removing the good of value $1$ from the other agent's bundle leaves value $2^h-2$. Hence, her factor is at most $(1/2)/(2^h-2)<\beta$, a contradiction.

\paragraph{Three or more agents.}
Fix $n\geq3$, and suppose for contradiction that some consistent algorithm guarantees $\beta$-EEFX for some $\beta>0$ when all agents' predictions are overestimates with error factor $3/2$. Let $q=3/2$, and let $h\geq4$ be an integer such that $q^{2-h}<\beta$. Let $\delta=q^{-h}$ and $Q=\{q^{-1},\ldots,q^{-h}\}$, so that the values in $Q$ sum to $T=2-2\delta$. We identify these distinct values with goods. No two disjoint nonempty subsets of $Q$ have the same sum. Indeed, otherwise, multiplying the equality of their sums by $q^h$ would yield a nonzero polynomial with coefficients in $\{-1,0,1\}$ having $q=3/2$ as a root; since its leading coefficient is $1$ or $-1$, this contradicts the rational root theorem.

Let $J>3$. Every agent's predicted multiset consists of $n-3$ copies of $J$, the value $1$, and the values in $Q$. The goods arrive with common values: first the goods of value $J$, then the goods of $Q$, and finally one last good. If the predictions were exact, the last good would have common value $1$. We show that in every EEFX allocation of this exact continuation, the goods of value $J$ are allocated to different agents, and the good of value $1$ forms the bundle of another agent.

We first show that there do not exist three disjoint bundles contained in $Q$ such that, for each of them, the goods of $\{1\}\cup Q$ outside the bundle can be partitioned into two parts toward which the bundle satisfies EFX; some goods of $Q$ may belong to none of the three bundles. An empty bundle cannot satisfy this condition, since there are $h+1>2$ goods of positive value to be partitioned into two parts. Suppose that one of the bundles has value at least $1$. At most one of the other two bundles contains the good of value $\delta$; let $X$ be one that does not, and let $x$ be its value. Since $T<2$, we have $x<1$. Hence, the good of value $1$ must form a part by itself in the partition for $X$, so the other part is $Q\setminus X$, whose smallest value is $\delta$. Thus $x\geq T-x-\delta$, i.e., $x\geq(T-\delta)/2=1-3\delta/2$. Together with the bundle of value at least $1$, this exceeds the total value $T=2-2\delta$ of $Q$, a contradiction. Consequently, all three bundles have value less than $1$.

At least two of the three bundles do not contain the good of value $\delta$, and by the same argument, each of them has value at least $(T-\delta)/2$. Hence, the third bundle and the goods of $Q$ outside the three bundles have total value at most $\delta$. Since the third bundle is nonempty and $\delta$ is the smallest value in $Q$, it consists of the good of value $\delta$ alone. Its partition must then consist of the parts $\{1\}$ and $Q\setminus\{\delta\}$, and the smallest value in the latter part is $q\delta$. The EFX condition would therefore require $\delta\geq T-\delta-q\delta$. However, for $h\geq4$,
\[
T-(q+2)\delta
=2-(q+4)\delta
\geq 2-\frac{11}{2}\left(\frac23\right)^4
=\frac{74}{81}>0,
\]
a contradiction. Thus, three such bundles do not exist.

Now consider an EEFX allocation of the exact continuation. An agent with no good of value $J$ has a bundle of value at most $1+T<3<J$, so in any partition showing that her bundle is EEFX for her, each good of value $J$ forms a part by itself, and the goods of $\{1\}\cup Q$ outside her bundle lie in the two remaining parts. If the $n-3$ goods of value $J$ had fewer than $n-3$ recipients, then at least four agents would have no good of value $J$, and at least three of them would not have the good of value $1$ either; their bundles would be three disjoint bundles contained in $Q$ satisfying the condition ruled out above. Hence, the goods of value $J$ have $n-3$ different recipients, and exactly three agents have no good of value $J$. The good of value $1$ must be allocated to one of these three agents, since otherwise their bundles would again be three disjoint bundles contained in $Q$ satisfying this condition.

Let $X$ and $Y$ be the bundles of the other two agents with no good of value $J$, let $x$ and $y$ be their values, and let $L$ be the total value of the goods in $Q\setminus(X\cup Y)$, so that $x+y+L=T$. Suppose that $L>0$. If $x\geq1$, then $y<1$, so in any partition showing that $Y$ is EEFX for its recipient, all goods of value $J$ and the good of value $1$ form parts by themselves, and the remaining part is $Q\setminus Y$. Since $Q\setminus(X\cup Y)$ is nonempty with total value $L$, the smallest value in $Q\setminus Y$ is at most $L$, so
\[
y\geq x+L-\min(Q\setminus Y)\geq x\geq1,
\]
a contradiction. By symmetry, both $x$ and $y$ are less than $1$, and the same reasoning gives
\[
x\geq y+L-\min(Q\setminus X)\quad\text{and}\quad
y\geq x+L-\min(Q\setminus Y).
\]
Let $\ell$ be the smallest value in $Q\setminus(X\cup Y)$. Adding the two inequalities yields
\[
2L\leq\min(Q\setminus X)+\min(Q\setminus Y)\leq2\ell\leq2L,
\]
so equality holds throughout. Since all values are positive, $\ell=L$ implies that $Q\setminus(X\cup Y)$ consists of a single good, and the two inequalities then imply $x=y$. If $x=y=0$, then $X$ and $Y$ are empty and $Q$ consists of a single good, contradicting $h\geq4$. Otherwise, $X$ and $Y$ are disjoint nonempty subsets of $Q$ with the same sum, a contradiction. Hence $L=0$.

It follows that all goods of $Q$ are allocated to these two agents, so each good of value $J$ and the good of value $1$ forms the bundle of a different agent. In particular, consistency forces the algorithm to leave some agent without any good before the last good arrives. Now let the last good have common value $u=q^{-(h+1)}$ instead of $1$. The goods of value $J$ are unchanged, and the true values of the other goods are exactly the predicted values $\{1\}\cup Q$ divided by $q$. Hence, the predictions are overestimates with error factor $3/2$. (For $n=3$, there are no goods of value $J$, and every true value is $2/3$ of the corresponding predicted value.)

The agent without any good before the last arrival must receive the last good for the allocation to be $\beta$-EEFX, since otherwise there are more than $n-1$ goods of positive value outside her bundle, so some part of any partition contains at least two of them. Her bundle then has value $u$. Since $u<\beta q^{-3}<\beta J$, in any partition showing that her bundle is $\beta$-EEFX for her, each good of value $J$ forms a part by itself, and the goods of $Q$ lie in the two remaining parts. Two of the three largest goods of $Q$ lie in the same part, and removing the larger of them leaves value at least $q^{-3}$. Hence, her factor is at most $u/q^{-3}=q^{2-h}<\beta$, a contradiction.

In both constructions, the exact and inaccurate continuations have the same number of goods and the same predictions, and the restrictions on the allocation of the first goods follow from consistency for every sequence of allocation decisions. Consequently, the arguments also apply to randomized algorithms.
\end{proof}

\end{document}